\documentclass[a4paper,UKenglish]{lipics}
\usepackage[bbgreekl]{mathbbol}
\usepackage{url}
\usepackage[usenames,dvipsnames]{xcolor}
\usepackage{amsfonts,amsmath,amssymb}
\usepackage{verbatim,paralist,mathtools}

\usepackage{tikz}
\usetikzlibrary{arrows,shapes,decorations}

\usepackage{syntax}
\usepackage{mdwlist}
\usepackage{multirow}

\usepackage{listings}
\usepackage{sty/mathematics}
\usepackage{sty/programs}
\usepackage{sty/proofsInAppendix}
\fixstatement{theorem}
\fixstatement{lemma}
\fixstatement{corollary}

\author[1]{Georgel Calin}
\author[2]{Egor Derevenetc}
\author[3]{Rupak Majumdar}
\author[1]{\newline{}\hspace*{0cm}\hfill Roland Meyer} 
\affil[1]{\vspace{-0.5cm}University of Kaiserslautern, \texttt{\{calin, meyer\}@cs.uni-kl.de}}
\affil[2]{Fraunhofer ITWM, \texttt{egor.derevenetc@itwm.fraunhofer.de}}
\affil[3]{MPI-SWS, \texttt{rupak@mpi-sws.org}}
\title{A Theory of Partitioned Global Address Spaces\footnote{The second author
was granted by the Competence Center High Performance Computing and
Visualization (CC-HPC)
of the Fraunhofer Institute for Industrial Mathematics (ITWM). The work was partially supported by the PROCOPE project ROIS:
Robustness under Realistic Instruction Sets.}}
\authorrunning{G. Calin, E. Derevenetc, R. Majumdar, and R. Meyer}

\usepackage{setspace}
\usepackage{soul}
\makeatletter
\DeclareRobustCommand*\myul{%
    \def\SOUL@everyspace{\underline{\space}\kern\z@}
    \def\SOUL@everytoken{%
     \setbox0=\hbox{\the\SOUL@token}%
     \ifdim\dp0>\z@
        \the\SOUL@token
     \else
        \underline{\the\SOUL@token}%
     \fi}
\SOUL@}

\begin{document}

\newif\ifappendix

\maketitle
\vskip -1em
\begin{abstract}
Partitioned global address space (PGAS) is a parallel programming model for
the development of high-performance applications on clusters.
It provides a global address space partitioned among the cluster nodes, and is
supported in programming languages like C, C++, and Fortran by means of APIs.
In this paper we provide a formal model for the semantics of single
instruction, multiple data programs using PGAS APIs.
Our model reflects the main features of popular real-world APIs such as
SHMEM, ARMCI, GASNet, GPI, and GASPI.

A key feature of PGAS is the support for one-sided communication:
a node may directly read and write the memory located at a remote node,
without explicit synchronization with the processes running on the
remote side.
One-sided communication increases performance by decoupling process
synchronization from data transfer, but requires the programmer to reason 
about appropriate synchronizations between reads and writes.
As a second contribution, we propose and investigate {\em robustness}, 
a criterion for correct synchronization of PGAS programs.
Robustness corresponds to acyclicity of a suitable
happens-before relation defined on PGAS computations.
The requirement is finer than the classical data race freedom and
rules out most false error reports.

Our main technical result is an algorithm for checking robustness of PGAS programs. 
The algorithm makes use of two insights.
Using combinatorial arguments we first show that, if a PGAS program is not robust,
then there are computations in a certain normal form that violate happens-before acyclicity. 
Intuitively, normal-form computations delay remote accesses in an ordered way.
We then devise an algorithm that checks for cyclic normal-form computations.
Essentially, the algorithm is an emptiness check for a novel automaton model
that accepts normal-form com\-pu\-ta\-tions in streaming fashion.
Altogether, we prove the robustness problem is $\pspace$-complete.
\end{abstract}

\section{Introduction}\label{Section:Introduction}

Partitioned global address space (PGAS) is a parallel programming model for the development of high-performance software on clusters.
The PGAS model provides a global address space to the programmer
that is partitioned among the cluster nodes (see Figure~\ref{Figure:onetooneRDMA}(b)).
Nodes can read and write their local memories, but additionally access the remote address space
through (synchronous or asynchronous) API calls.
PGAS is a popular programming model, and supported by many PGAS APIs,
such as SHMEM~\cite{chapman2010introducing}, ARMCI~\cite{nieplocha1999armci},
GASNET \cite{bonachea2002gasnet}, GPI \cite{machado2009fraunhofer}, and GASPI~\cite{GASPI},
as well as by languages for high-performance computing, such as
UPC~\cite{UPC}, Titanium~\cite{hilfinger2005titanium}, and Co-Array Fortran~\cite{numrich1998co}.

A key ingredient of PGAS APIs is their support for one-sided communication.
Unlike in traditional message passing interfaces,
a node may directly read and write the memory located at a remote node without
explicit synchronization with the remote side.
One-sided com\-mu\-nication can be efficiently implemented on top of networking
hardware featuring remote direct memory access (RDMA),
and increases performance of PGAS programs by avoiding
unnecessary synchronization between the sender and the
receiver \cite{machado2009fraunhofer,dinanimplementation}.

However, the use of one-sided communication introduces additional non-determinism
in the ordering of memory reads and writes, and makes reasoning about program correctness
harder.
Figure~\ref{Figure:onetooneRDMA}(a) demonstrates a subtle bug arising out of
improper synchronizations: while the barriers ensure all processes are at the same
control location, the remote writes may or may not have completed when address
$\lit*y$ is accessed after the barrier.

\begin{figure}[t]
\begin{minipage}[b]{0.40\textwidth}
\begin{lstlisting}
int x = 1, y = 0;
write(x, rightNeighbourRank,
      y, myWriteQ);
barrier();
assert(y == 1);
\end{lstlisting}
\vspace{0.25cm}
\end{minipage}
\hspace{0.35cm}
\begin{minipage}[b]{0.55\textwidth}
\centering
\vspace{-1.5cm} 
\begin{tikzpicture}[scale=0.75,transform shape]
	\node (nic1) at (2.375,0.8) [draw,minimum height=0.5cm,shape=rectangle] {\small NIC};
	\node (s1) at (1,2.4) [draw,minimum height=0.5cm,shape=rectangle] {\small Shared Memory $1\times\addrdomain$};
	\node (p1) at (-0.095,0.8) [draw,minimum height=0.5cm,shape=rectangle] {\small Process 1};
	\node (local1) at (0.85,1.75) [draw,minimum height=0.5cm,shape=rectangle] {\small Local Registers};

	\draw[dash pattern=on 1pt off 0.9pt] node [label=above left:Node 1] at (1,1.55) [draw,minimum width=4.2cm, minimum height=2.4cm,shape=rectangle] {};	

	\draw[stealth'-stealth'] (p1) -- (nic1);
	\draw[stealth'-stealth'] (nic1) to (s1.349);
	\draw[stealth'-stealth'] (p1.156) to (s1.189);
	\draw[stealth'-stealth'] (p1.34) to (local1.205);
			
	\node (dots) at (3.75,1.75) {$\bullet\bullet\bullet$};

	\node (nic0) at (7.875,0.8) [draw,minimum height=0.5cm,shape=rectangle] {\small NIC};
	\node (s0) at (6.5,2.4) [draw,minimum height=0.5cm,shape=rectangle] {\small Shared Memory $\mathrm{N}\times\addrdomain$};
	\node (p0) at (5.405,0.8) [draw,minimum height=0.5cm,shape=rectangle] {\small Process N};
	\node (local0) at (6.35,1.75) [draw,minimum height=0.5cm,shape=rectangle] {\small Local Registers};
	
	\draw[dash pattern=on 1pt off 0.9pt] node [label=above left:Node N] at (6.5,1.55) [draw,minimum width=4.2cm, minimum height=2.4cm,shape=rectangle] {};	

	\draw[stealth'-stealth'] (p0) -- (nic0);
	\draw[stealth'-stealth'] (nic0) to (s0.349);
	\draw[stealth'-stealth'] (p0.156) to (s0.189);
	\draw[stealth'-stealth'] (p0.34) to (local0.205);
	
	\draw[stealth'-stealth'] (nic1) |- (3,0.1) -| node [below] {\small Network} (nic0);
\end{tikzpicture}
\end{minipage}

\captionof{figure}{(a) Program \onetoone{} is the \emph{compute and exchange results} idiom often found in PGAS applications.
Each process copies an integer value to its neighbour.
\texttt{write} asks the hardware to copy the value of address $x$ to $y$ on the right neighbouring node.
\texttt{barrier} blocks until all processes reach the barrier.
The assertion can fail, as the \texttt{barrier}
may execute before the \texttt{write} completes.
(b) PGAS architecture --- NIC stands for network interface controller.}
\label{Figure:onetooneRDMA}
\end{figure}
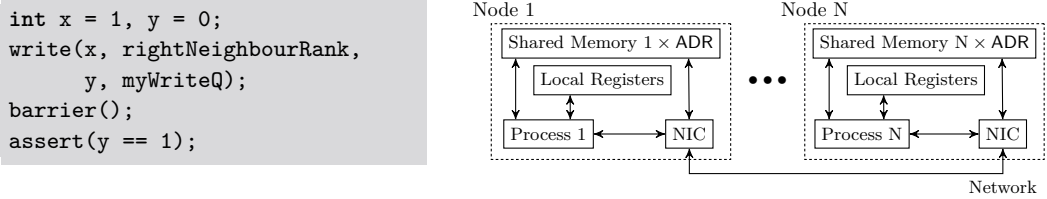

We make two contributions in this paper.

First, we provide a core calculus of PGAS APIs that models concurrent processes
sharing a global address space and accessing remote memory through one-sided reads and writes.
Despite the popularity of PGAS APIs in the high-performance computing community,
to the best of our knowledge, there are no formal models for common PGAS APIs.

Second,
we define and study a correctness criterion called \emph{robustness} for PGAS programs.
To understand robustness,
we begin with a classical and intuitive correctness condition, {\em sequential consistency}~\cite{Lamport79}.
A computation is sequentially consistent if its memory accesses happen atomically and in the order
in which they are issued.
Sequential consistency is too strong a criterion for PGAS programs,
where time is required to access remote memory and accesses themselves can be reordered.
Robustness is the weaker notion that all computations of the program have the same
happens-before (data and control) dependencies \cite{ShashaSnir88}
as some sequentially consistent computation.
Our notion of robustness captures common programming error patterns \cite{dice09:park,Muzahid2012},
and is derived from a similar notion in shared memory multiprocessing \cite{ShashaSnir88}.
Related correctness criteria have been proposed for
weak memory models~\cite{burckhardt-musuvathi-CAV08,Owens2010,Alglave2010,AlglaveM11,BMM11,Sen2011,bouajjani2013checking}.

A simpler correctness property would be {\em data race freedom} (DRF), in
which no two processes access the same address at the same time, with at least one access being a write~\cite{AdveHillDRF1993}.
Indeed, data race free programs are sequentially consistent.
Unfortunately, DRF is too strong a requirement in practice~\cite{park2011efficient}, and leads to numerous false alarms.
Many common synchronization idioms for PGAS programs, such as producer-consumer synchronization,
and many concurrent data structure implementations,
contain benign data races.
Instead, the notion of robustness captures the intuitive requirement that, even when events are
reordered in a computation, there are no causality cycles.
Our notion of causality is the standard {\em happens-before} relation from~\cite{ShashaSnir88}.

We study the algorithmic verification of robustness.
Our main result is that robustness is decidable (actually $\pspace$-complete)
for PGAS programs, assuming a finite data domain and finite memory.
Note that our model of PGAS programs is infinite-state even when the data domain is finite:
one-sided communication allows unboundedly many requests to be in flight simultaneously (a feature modeled in our formalism
using unbounded queues).

Our decidability result uses two technical ingredients.
First, we show that among all computations violating robustness, there is always
one in a certain normal form.
The normal form partitions the violating computation into phases:
the first phase initiates memory reads and writes, and the latter phases complete the
reads and writes in the same order in which they were initiated.

Second, we provide an algorithm to detect violating computations in this normal form.
We take a language-theoretic view, and introduce a multiheaded automaton model which
can accept violating computations in normal form.
Then the problem of checking robustness reduces to checking emptiness for multiheaded automata.
Interestingly, since the normal form maintains orderings of accesses,
the multiple heads can be exploited to accept violating computations without explicitly modeling
unbounded queues of memory access requests.
The resulting class of languages contains non-context-free ones (such as $a^nb^nc^n$), but retains sufficient decidability properties.
Altogether this yields a $\pspace$ decision procedure for checking robustness of programs using PGAS APIs.

For lack of space, full constructions and proofs are given in the appendix.

\smallskip
\noindent\textbf{Related Work}
Although PGAS APIs are popular in the high-performance computing
community~\cite{bonachea2002gasnet,chapman2010introducing,GASPI,machado2009fraunhofer,nieplocha1999armci},
to the best of our knowledge,
no previous work provides a unifying formal semantics that incorporates
one-sided asynchronous communication.
As for synchronization correctness, only recently Park et al.
proposed a testing framework for data race detection and implemented it for the UPC language~\cite{park2011efficient}.
However, the authors argue that many data races are actually not harmful,
a claim they support through the analysis of the NAS Parallel Benchmarks~\cite{UPCNPB}.
For this reason, in contrast to data race freedom~\cite{AdveHillDRF1993},
we consider robustness as a more precise notion of appropriate synchronization.

The robustness problem was posed by Shasha and Snir~\cite{ShashaSnir88} for shared memory multiprocessing.
They showed that non sequentially consistent computations have a happens-before cycle.
Alglave and Maranget \cite{Alglave2010,AlglaveM11} extended this result. They developed a general theory for reasoning about robustness problems, even among different architectures.
Owens \cite{Owens2010} proposed a notion of appropriate synchronization that is based on triangular data races.
Compared to robustness, triangular race freedom requires heavier synchronization, which is undesirable for performance reasons.

We consider here the algorithmic problem of checking robustness.
For programs running on weak memory models the problem has been addressed in~\cite{burckhardt-musuvathi-CAV08,Sen2011,AlglaveM11}, but none of these works provides a (sound and complete) decision procedure.
The first complete algorithm for checking robustness of programs running on Total Store Ordering (TSO) architectures was given in~\cite{BMM11}.
It is based on the following locality property.
If a TSO program is not robust, then there is a violating computation where only one process delays commands.
This insight leads to a reduction of robustness to reachability in the sequential consistency model~\cite{bouajjani2013checking}.
PGAS programs allow more reorderings than TSO ones and,
as a consequence, locality does not hold. Instead, our decision procedure relies on
a complex normal form for computations and on a sophisticated automata-theoretic algorithm to look for normal-form violations.

\section{PGAS Programs}\label{Section:Programs}

\subsection{Features of PGAS Programs}
PGAS programs are \emph{single instruction, multiple data} programs running on a cluster (see Figure~\ref{Figure:onetooneRDMA}(b)).
At run time, a PGAS program consists of multiple processes executing the same code on different nodes.
Each process has a \emph{rank}, which is the index of the node it runs on.
The processes can access a global address space partitioned into local address spaces for each process.
Local addresses can be accessed directly.
Remote addresses  (addresses belonging to different processes) are accessed using API calls, which come in different flavors.

SHMEM~\cite{chapman2010introducing} provides synchronous remote reads where the invoking process
waits for completion of the command.
Remote write commands are asynchronous, and no ordering is guaranteed between writes, even to the same remote node.
The ordering can, however, be enforced by a special fence command.

ARMCI~\cite{nieplocha1999armci} features synchronous as well as asynchronous read and write commands.
The asynchronous variants of the commands return a handle that can be waited upon.
When the wait on a read handle is over, the data being read has arrived and is accessible.
When the wait on a write handle is over,
the data being written has been sent to the network but might not have reached its destination.
Unlike operations to different nodes, operations to the same remote node are executed in their issuing order.

GASNet~\cite{bonachea2002gasnet}, like ARMCI, provides both synchronous and asynchronous versions of reads and writes.
Commands return a handle that can be waited upon, and a return from a wait implies full completion of the operation.
The order in which asynchronous operations complete is intentionally left unspecified.

GPI~\cite{machado2009fraunhofer} and GASPI~\cite{GASPI} only support asynchronous read and write commands.
Each read or write operation is assigned a queue identifier.
In GPI, operations with the same queue id and to the same remote node are executed in the order in which they were issued; in GASPI this guarantee does not hold.
One can wait on a queue id, and the wait returns when all commands in the queue are fully completed, on both the local and the remote side.

Summing up, in a uniform PGAS programming model it should be possible to
\begin{asparaitem}
\item perform synchronous and asynchronous data transfers,
\item assign an asynchronous operation a handle or a queue id,
\item wait for completion of an individual command or of all commands in a given queue,
\item enforce ordering between operations.
\end{asparaitem}
We define a core model for PGAS that supports all these features.
Our model only uses asynchronous remote reads and writes
with explicit queues, but is flexible
enough to accommodate all the above idioms.

\subsection{Syntax of PGAS Programs}

We define PGAS programs and their semantics in terms of automata.
A (non-deterministic) \emph{automaton} is a tuple
$\automaton=(\states,\alphabet,\transitions,\initialstate,\finalstates)$,
where $\states$ is a set of states,  $\alphabet$ is a finite alphabet,
$\transitions\subseteq\states\times(\alphabet\cup\set{\emptysequence})\times \states$
is a set of transitions,
$\initialstate\in\states$ is an initial state,
and $\finalstates\subseteq\states$ is a set of final states.
We call the automaton \emph{finite} if the set of states is finite.
We write $\state_1\transitionto{a}\state_2$ if $(\state_1,a,\state_2)\in\transitions$, and extend
the relation to computations $\sigma\in\alphabet^*$ in the expected way.
The \emph{language} of the automaton is
$\langOf{\automaton}:=\setcond{\sigma\in\alphabet^*}{\initialstate\transitionto{\sigma}\state\text{ for some }\state\in\finalstates}$.
We write $\length{\sigma}$ for the length of a computation $\sigma\in\Sigma^*$,
and use $\succorder{\sigma}$ to denote the successor
relation among the letters in $\sigma$.
We write $a\before{\sigma}b$ if $\sigma=\sigma_1\cdot{}a\cdot\sigma_2\cdot{}b\cdot\sigma_3$ for some $\sigma_1,\sigma_2,\sigma_3\in\Sigma^*$.

A \emph{PGAS program} $(\program,\nodecount)$ consists of a program code $\program$ and a fixed number $\nodecount \geq 1$ of cluster nodes.
The program code $\program:=(\controlstates,\commands,\instructions,\initialcontrolstate,\controlstates)$ is a finite automaton with a set of control states $\controlstates$, all of them are final, initial state $\initialcontrolstate$, and a set of transitions $\instructions$ labeled with \emph{commands} $\commands$.

Let $\datadomain$, $\addrdomain$, and $\queuedomain$ be finite domains of values (containing a value $0$), addresses, and queue identifiers, respectively.
Let $\regdomain$ be a finite set registers that take values from $\datadomain$.
The grammar of commands is given in Figure~\ref{Figure:SyntaxCommands}.
For simplicity, we will assume $\datadomain = \addrdomain = \queuedomain$.
The set of expressions is defined over constants from $\datadomain$, registers from $\regdomain$, and (unspecified) operators over $\datadomain$.
The set of commands $\commands$ includes local assignments and conditionals (\lit*{assume}),
remote read and write API calls $\texttt{read}$ and $\texttt{write}$ respectively, and barriers $\texttt{barrier}$.

At run time, there is a process on each node $\overline{1,\nodecount}$ that executes program $\program$, where $\overline{M,N}:=\set{M,M+1,\ldots, N}$.
We will identify each process with its rank from $\ranks:=\overline{1,\nodecount}$.
For modeling purposes, one may assume there are special constant expressions
that let a process learn about its rank in $\ranks$ and about the total number of processes $\nodecount$.

\begin{figure}[t]
\vskip -2.5em
\begin{minipage}[b]{0.6\textwidth}
\begin{small}
\setlength{\grammarindent}{2em}
\setlength{\grammarparsep}{\parskip}
\begin{grammar}
  <cmd> ::= <reg> $\leftarrow$ "mem["<expr>"]"
  \alt "mem["<expr>"]" $\leftarrow$ <expr>
  \alt <reg> $\leftarrow$ <expr>
  \alt "assume("<expr>")"
  \alt "read("<local-adr>","<rank>","<remote-adr>","<que-id>")"
  \alt "write("<local-adr>","<rank>","<remote-adr>","<que-id>")"
  \alt "barrier"

\end{grammar}
\end{small}
\captionof{figure}{
  Syntax of commands. 
\synt{reg} ranges over $\regdomain$; 
expressions \synt{expr}, 
local addresses \synt{local-adr}, 
remote addresses \synt{remote-adr}, and
queue identifiers \synt{que-id} range over expressions;
ranks \synt{rank} over $\overline{1,\nodecount}$-valued expressions.
}
\label{Figure:SyntaxCommands}
\end{minipage}
\hfill
\begin{minipage}[b]{0.37\textwidth}
\centering
\begin{tikzpicture}[nodes={rectangle,draw=none,fill=none}]
    \node (write0) {$\nodeone{\writekind}$}; 
    \node [node distance=2.5cm,right of=write0] (write1) {$\nodetwo{\writekind}$}; 
    \node [node distance=0.75cm,below of=write0] (popa0) {$\nodeone{\popakind}$}; 
    \node [node distance=0.75cm,below of=write1] (popa1) {$\nodetwo{\popakind}$}; 
    \node [node distance=0.75cm,below of=popa0] (popb0) {$\nodeone{\popbkind}$}; 
    \node [node distance=0.75cm,below of=popa1] (popb1) {$\nodetwo{\popbkind}$}; 
    \node [node distance=0.75cm,below of=popb0] (barrier0) {$\nodeone{\barrierkind}$}; 
    \node [node distance=0.75cm,below of=popb1] (barrier1) {$\nodetwo{\barrierkind}$}; 
    \node [node distance=0.75cm,below of=barrier0] (load0) {$\nodeone{\loadkind}$}; 

  \draw[<->] (write0)   edge node [midway,right] {$\scriptstyle \equivalence$} (popa0);
  \draw[<->] (popa0)    edge node [midway,right] {$\scriptstyle \equivalence$} (popb0);
  \draw[->]  (write0)   edge[bend right=40] node [very near end,left] {$\scriptstyle \po$} (barrier0);
  \draw[->]  (barrier0) edge node [pos=0.7,left] {$\scriptstyle \po$} (load0);

  \draw[<->] (write1)   edge node [midway,left] {$\scriptstyle \equivalence$} (popa1);
  \draw[<->] (popa1)    edge node [midway,left] {$\scriptstyle \equivalence$} (popb1);
  \draw[->]  (write1)   edge[bend left=40] node [very near end,right] {$\scriptstyle \po$} (barrier1);

  \draw[->]  (load0)    edge node [pos=0.85,below] {$\scriptstyle \cf$} (popb1.west);
  \draw[<->] (barrier0) edge [bend right=15] node [very near end, below] {$\scriptstyle \equivalence$} (barrier1);
\end{tikzpicture}
\vskip 0.5em
\captionof{figure}{Happens-before relation of $\tau_{\onetoone}$ (Example~\ref{Example:onetooneComputation}).
Computation $\tau_{\onetoone}$ violates robustness.}
\label{Figure:onetooneTrace}
\end{minipage}

\end{figure}

\subsection{Semantics of PGAS Programs}
The semantics of a PGAS program $(\program, \nodecount)$ is defined using a \emph{state-space automaton}
$\eventautomaton(\program, \nodecount):=(\eventstates,\events,\eventtransitions, \initialeventstate,\finaleventstates)$.
A state $\state\in\eventstates$ is a tuple
$\state=(\stateconf,\memoryconf,\queueconfa,\queueconfb)$, where
state configuration $\stateconf\colon\ranks\to\controlstates$ maps each process to its current control state,
memory configuration $\memoryconf\colon\ranks\times(\regdomain\cup\addrdomain)\to\datadomain$ maps each process
to the values stored in each register and at each address,
queue configuration
$\queueconfa\colon\ranks\times\queuedomain\to(\ranks\times\addrdomain\times \ranks\times\addrdomain)^*$
maps each process to remote read and write requests that were issued,
and
$\queueconfb\colon\ranks\times\queuedomain\to(\ranks\times\addrdomain\times \datadomain)^*$
contains values to be transferred.

The initial state is
$\initialeventstate:=(\initialstateconf,\initialmemoryconf,\initialqueueconfa, \initialqueueconfb)$,
where for all ranks $\rank\in\ranks$, registers and addresses
$\anaddr\in\regdomain\cup \addrdomain$, and queue identifiers
$\queueid\in\queuedomain$, we have
$\initialstateconf(\rank):=\initialcontrolstate$,
$\initialmemoryconf(\rank,\anaddr):=0$, and
$\initialqueueconfa(\rank,\queueid):=\emptysequence=:\initialqueueconfb(\rank, \queueid)$.
The set of final states is
$\finaleventstates:=\setcond{(\stateconf,\memoryconf,\queueconfa,
\queueconfb)\in\eventstates}{\queueconfa(\rank, \queueid)=\emptysequence = \queueconfb(\rank, \queueid)\text{ for all }\rank\in\ranks, \queueid\in\queuedomain}$.
The semantics of commands ensures queues can always be emptied,
so acceptance with empty queues is not a restriction.

The alphabet of $\eventautomaton(\program, \nodecount)$ is the set of
\emph{events} $\events:=\kinddomain\times\ranks\times((\ranks\times\addrdomain)\cup\set{\bottom})$ with
\emph{event kinds} $\kinddomain:=\set{\loadkind,\storekind,\assignkind,\assumekind,\readkind,\writekind ,\popakind,\popbkind,\barrierkind}$.
Consider an event $\event=(\kind,\rank,(\rank_{\anaddr},\anaddr))\in\events$.
We use $\kindOf{\event}=\kind$ to determine the kind of the event,
$\rankOf{\event}=\rank$ for the rank of the process that produced the event, and
$\addrOf{\event}=(\rank_{\anaddr},\anaddr)$ to obtain the rank and the address that are \emph{accessed} by the event.
If $\kindOf{\event}\in\set{\loadkind,\popakind}$, then $\event$ is said to be a \emph{read of $(\rank_{\anaddr},\anaddr)$}.
If $\kindOf{\event}\in\set{\storekind,\popbkind}$, then $\event$ is a \emph{write of} address $\addrOf{\event}$.

Table~\ref{Table:EventAutomatonRules} shows a subset of the transition relation $\eventtransitions$;
other rules are similar.
When a process executes a remote write command, Rule~($\writekind$), a
new item is added to a queue in $\queueconfa$.
This item contains the source rank and source address from which the data will be copied,
together with the destination rank and destination address to which the data will be copied.
Eventually, the item is popped from the queue in $\queueconfa$,
Rule~($\popakind$), the value is read from the source address, and a new item is pushed into the corresponding
queue in $\queueconfb$.
The new item contains the destination rank and destination address, and the value that was read from the source address.
Eventually, this item is popped from the queue, Rule~($\popbkind$),
and the value is written to the destination address in the destination rank.
Modeling two queue configurations yields a symmetry between remote writes and reads:
a read can be interpreted as a write that comes upon request.
Moreover, two queue configurations capture well the delays
between request creation, reading of the data, and writing of the data.

The semantics of a PGAS program $\computationsOf{\program, \nodecount}:=\langOf{\eventautomaton(\program, \nodecount)}\subseteq \events^*$
is the set of computations of the state-space automaton.

\begin{table}[t!]
\centering

\therules{
\ifappendix
  \therule
  {$\command = \thestore{\anexpr_{\anaddr}}{\anexpr_{\aval}}$}
  {$\state
    \transitionto{(\storekind,\rank,(\rank,\valueOf{\anexpr_{\aval}}))}
   (\stateconf',\memoryconf[(\rank,\valueOf{\anexpr_{\anaddr}}):=\valueOf{\anexpr_{\aval}}],\queueconfa,\queueconfb)$}
  {($\storekind$)}

  \therule
  {$\command = \theassign{\areg}{\anexpr}$}
  {$\state
    \transitionto{(\assignkind,\rank,\bottom)}
    (\stateconf',\memoryconf[(\rank,\areg):=\valueOf{\anexpr}],\queueconfa,\queueconfb)$}
  {($\assignkind$)}

  \therule
  {$\command = \theassert{\anexpr}$\qquad $\valueOf{\anexpr} \neq 0$}
  {$\state
    \transitionto{(\assertkind,\rank,\bottom)}
    (\stateconf',\memoryconf,\queueconfa,\queueconfb)$}
  {($\assertkind$)}

  \therule
  {$\command=\theread{\anexpr_{\anaddr}^{\local}}{\anexpr_{\rank}^{\remote}}{\anexpr_{\anaddr}^{\remote}}{\anexpr_{\queueid}}$
   \qquad
   $\queueconfa(\rank,\valueOf{\anexpr_{\queueid}}) = \alpha$}
  {$\state
    \transitionto{(\readkind,\rank,\bottom)}(\stateconf',\memoryconf,\queueconfa[(\rank,\valueOf{\anexpr_{\queueid}}):=\alpha\cdot(\valueOf{\anexpr_{\rank}^{\remote}},\valueOf{\anexpr_{\anaddr}^{\remote}},\rank,\valueOf{\anexpr_{\anaddr}^{\local}})],\queueconfb)$}
  {($\readkind$)}

\else

  \therule
  {$\command = \theload{\areg}{\anexpr_{\anaddr}}$}
  {$\state
    \transitionto{(\loadkind,\rank,(\rank,\valueOf{\anexpr_{\anaddr}}))}
    (\stateconf',\memoryconf[(\rank,\areg):=\memoryconf(\rank,\valueOf{\anexpr_{\anaddr}})],\queueconfa,\queueconfb)$}
  {($\loadkind$)}

  \therule
  {$\command=\thewrite{\anexpr_{\anaddr}^{\local}}{\anexpr_{\rank}^{\remote}}
{\anexpr_{\anaddr}^{\remote}}{\anexpr_{\queueid}}$
\qquad   $\queueconfa(\rank,\valueOf{\anexpr_{\queueid}})=\alpha$}
  {$\state
    \transitionto{(\writekind,\rank,\bottom)}
(\stateconf',\memoryconf,\queueconfa[(\rank,\valueOf{\anexpr_{\queueid}}
):=\alpha\cdot(\rank,\valueOf{\anexpr_{\anaddr}^{\local}},\valueOf{\anexpr_{
\rank}^{\remote}}, \valueOf{\anexpr_{\anaddr}^{\remote}})],\queueconfb)$}
  {($\writekind$)}

  \therule
  {$\queueconfa(\rank,\queueid)=(\rank_{\source},\anaddr_{\source},\rank_{\destination},\anaddr_{\destination})\cdot\alpha$
  \qquad
  $\queueconfb(\rank,\queueid)=\beta$}
  {$\state
    \transitionto{(\popakind,\rank,(\rank_{\source},\anaddr_{\source}))}
    (\stateconf,\memoryconf,\queueconfa[(\rank,\queueid):=\alpha],\queueconfb[(\rank,\queueid):=\beta\cdot(\rank_{\destination},\anaddr_{\destination},\memoryconf(\rank_{\source},\anaddr_{\source}))])$}
  {($\popakind$)}

  \therule
  {$\queueconfb(\rank,\queueid)=(\rank_{\destination},\anaddr_{\destination},\aval)\cdot\beta$}
  {$\state
    \transitionto{(\popbkind,\rank,(\rank_{\destination},\anaddr_{\destination}))}
    (\stateconf,\memoryconf[(\rank_{\destination},\anaddr_{\destination}):=\aval],\queueconfa,\queueconfb[(\rank,\queueid):=\beta])$}
  {($\popbkind$)}

  \therule
  {$\stateconf(\rank)\transitionto{\thebarrier}\stateconf'(\rank)$ for each $\rank\in\ranks$}
  {$\state
    \transitionto{(\barrierkind,1,\bottom)\cdot(\barrierkind,2,\bottom)\cdots(\barrierkind,\nodecount,\bottom)}
    (\stateconf',\memoryconf,\queueconfa,\queueconfb)$}
  {($\barrierkind$)}
  \fi
}
\caption{Transition rules for $\eventautomaton(\program,\nodecount)$, given $\controlstate_1\transitionto{\command}\controlstate_2$ and 
current state $\state=(\stateconf,\memoryconf,\queueconfa, \queueconfb)$ with $\stateconf(\rank)=\controlstate_1$.
We set $\stateconf' := \stateconf[\rank:=\controlstate_2]$ to update $\stateconf$ so that process $\rank$ is at $\controlstate_2$.
$\valueOf{\anexpr}$ denotes the evaluation of expression $\anexpr$ in  memory configuration $\memoryconf$.
}
\label{Table:EventAutomatonRules\ifappendix Full\fi}
\end{table}

\begin{example}\label{Example:onetooneComputation}
Consider PGAS program $(\onetoone{}, 2)$ with the program code from Figure~\ref{Figure:onetooneRDMA}(a) being run on two nodes.
It has the following computation:
\begin{align*}
\tau_{\onetoone}=\nodeone{\writekind}\cdot\nodetwo{\writekind}\cdot\nodetwo{\popakind}\cdot\nodeone{\popakind}\cdot\nodeone{\barrierkind}\cdot\nodetwo{\barrierkind}\cdot\nodeone{\loadkind}\cdot\nodetwo{\popbkind}\cdot\nodeone{\popbkind}.
\end{align*}
Bold events belong to the process with rank $2$, the other events to the process with rank $1$. We have $\addrOf{\nodeone{\popakind}}=(1,x)$, $\addrOf{\nodeone{\popbkind}}=(2,y)$. Symmetrically, $\addrOf{\nodetwo{\popakind}}=(2,x)$ and $\addrOf{\nodetwo{\popbkind}}=(1,y)$.
The \texttt{assert} in Figure~\ref{Figure:onetooneRDMA} is a shortcut for a combination of load and assume, and in this computation $\addrOf{\loadkind}=(1, y)$.
\end{example}

\subsection{Simulating PGAS APIs}
Our formalism natively supports asynchronous data transfers and queues.
Operations in the same queue are completed in the order in which they were issued.
Using this, we can model the ordering guarantees given by ARMCI and GPI ---
by putting ordered operations into the same queue.

To model waiting on individual operations (waiting on a handle),
we associate a shadow memory address with each operation.
Before issuing the operation, the value at this address is set to $0$.
When the operation has been issued, the process sends to the same queue a read request which overwrites the shadow memory to $1$.
Now waiting on the individual operation can be implemented by polling on the shadow
address associated with the operation.
Waiting on all operations in a given queue is done similarly.
Synchronous data transfers are modeled by asynchronous transfers,
immediately followed by a wait.

\section{Robustness: A Notion of Appropriate Synchronization}\label{Subsection:Traces}
We now define \emph{robustness}, a correctness condition for PGAS programs.
Robustness is a weaker criterion than requiring all computations to be sequentially consistent~\cite{Lamport79}: it
allows for reordering of events as long as there are no causality cycles. 
As causality relation, we adopt the \emph{happens-before relation}~\cite{ShashaSnir88}.
Fix a computation $\tau\in\computationsOf{\program,\nodecount}$. 
Its happens-before relation is the union of the three 
relations we define next, 
$\happensbeforeOf{\tau} :=\ \progorder\cup\conflictorder\cup\equivalenceorder$.

The \emph{program order relation} $\progorder$ is the union of the program order relations for all processes: $\progorder\ :=\ \bigcup_{\rank\in\ranks}\progorder^{\rank}$.
Relation $\progorder^{\rank}$ gives the order in which events were issued in process $\rank$.
Formally, let $\tau'$ be the subsequence of all events $\event$ in $\tau$ such that $\rankOf{\event}=\rank$ and $\kindOf{\event}\not\in\set{\popakind,\popbkind}$.
Then $\progorder^{\rank}\ :=\succorder{\tau'}$.

The \emph{conflict relation} $\conflictorder$ orders conflicting accesses to the same address.
Let $\tau=\alpha\cdot\event_1\cdot\beta\cdot\event_2\cdot\gamma$, where $\event_1$ and $\event_2$ access the same address, and at least one of them is a write: $\addrOf{\event_1}=\addrOf{\event_2}=(\rank,\anaddr)$, $\kindOf{\event_1}\in\set{\storekind,\popbkind}$ or $\kindOf{\event_2}\in\set{\storekind,\popbkind}$.
If there is no $\event\in\beta$ such that $\addrOf{\event}=(\rank,\anaddr)$ and $\kindOf{\event}\in\set{\storekind,\popbkind}$, then $\event_1\conflictorder\event_2$.

The \emph{identity relation} $\equivalenceorder$ identifies events corresponding to the same command.
Let $\event$ be a remote read or write event, $\kindOf{\event}\in\set{\readkind,\writekind}$, and $\event_1$ and $\event_2$ be the corresponding requests, $\kindOf{\event_1}=\popakind$ and $\kindOf{\event_2}=\popbkind$.
Then we have $\event\equivalenceorder\event_1\equivalenceorder\event_2$.
In a similar way, $\equivalenceorder$ identifies matching barrier events in different processes.

We say a computation $\tau$ is \emph{violating} if the associated happens-before relation contains a non-trivial cycle, i.e., a cycle that is not included in $\equivalenceorder$.
Violating computations violate sequential consistency.
The robustness problem amounts to proving the absence of violations.
\begin{quote}
{\bf ROB}\hspace{0.45cm} Given a program $(\program, \nodecount)$, show that no computation $\tau\in\computationsOf{\program, \nodecount}$ is violating.
\end{quote}

\begin{example}\label{Example:onetooneTrace}
The happens-before relation of computation $\tau_{\onetoone}$ is depicted in Figure~\ref{Figure:onetooneTrace}.
It is cyclic, therefore $\tau_{\onetoone}$ is violating and $(\onetoone, 2)$ is not robust.
Indeed, no sequentially consistent execution of \onetoone{} allows the 
\texttt{assert} statements to load the initial value of $y$.
\end{example}
Our main result is the following.

\begin{theorem}\label{th:main}
{\bf ROB} is $\pspace$-complete.
\end{theorem}

The $\pspace$ lower bound follows from $\pspace$-hardness of control state reachability in sequentially consistent programs~\cite{Kozen77}.
To reduce to robustness, we add an artificial happens-before cycle starting in the target control state.
The rest of the paper shows a $\pspace$ algorithm, and hence upper bound, for the problem.

\section{Normal-Form Violations}\label{Section:WellShapedComputations}
We show that a PGAS program is not robust if and only if it has a violating computation of the following normal form.
\begin{definition}\label{Definition:WellShapedComputation}
Computation $\tau=\tau_1\cdot\tau_2\cdot\tau_3\cdot\tau_4\in\computationsOf{\program, \nodecount}$ is \emph{in normal form} if all $\event\in\tau_2\cdot\tau_3\cdot\tau_4$ satisfy $\kindOf{\event}\in\set{\popakind,\popbkind}$ and for all $a,b\in\tau_1$ 
with $\kindOf{a}, \kindOf{b}\notin\set{\popakind,\popbkind}$
and all $a',b'\in\tau_i$ with $i\in\parts$ we have:
\begin{align}
 a\before{\tau_1}b, a\not \equivalenceorder^* b, a\equivalenceorder^* a', b\equivalenceorder^* b'\quad \text{implies}\quad a'\before{\tau_i}b'.\tag{NF}\label{Equation:WellShaped}
\end{align}
\end{definition}
We explain the normal-form requirement~\eqref{Equation:WellShaped}.
Consider two accesses $a$ and $b$ to remote processes that can be found in the first part of the computation $\tau_1$.
Assume corresponding pop events $a'$ and $b'$ are delayed and can both be found in a later part of the computation, say $\tau_2$.
Then the ordering of $a'$ and $b'$ in $\tau_2$ coincides with the order of $a$ and $b$ in $\tau_1$. Computation $\tau_{\onetoone}$ is not in normal-form whereas $\tau\doubleprime_{\onetoone}$ in Figure~\ref{Figure:onetooneCycleOnComputation} is.
The following theorem guarantees that, in case of non-robustness, normal-form violations always exist.
\begin{theorem}\label{Theorem:OnlyNeedWellShapedComputations}
A PGAS program $(\program, N)$ is robust iff it has no normal-form violation.
\end{theorem}
Phrased differently, to decide robustness  our procedure should look for
normal-form violations.
The remainder of the section is devoted to proving Theorem~\ref{Theorem:OnlyNeedWellShapedComputations}.
We make use of the following property of PGAS programs:
every computation contains an event that can be deleted, in the sense that the result is again a computation of the program.
\begin{lemma}[Cancellation]\label{Lemma:Cancellation}
Consider $\emptysequence\neq \tau\in\computationsOf{\program,\nodecount}$.
There is an event $\event\in\tau$ so that $\tau\setminus\event\in \computationsOf{\program,\nodecount}$.
Computation $\tau\setminus\event$ is defined to remove $\event$ and all $\equivalenceorder$-related events from $\tau$.
\end{lemma}
\begin{proof}
Take as $\event$ the last event in $\tau$ with $\kindOf{\event}\not\in\set{\popakind,\popbkind}$.
All events to the right of $\event$ are unconditionally executable.
Moreover, $\tau$ does not have $\progorder$-successors following $\event$.
Therefore, the resulting computation $\tau\setminus\event$ is in $\computationsOf{\program, \nodecount}$.
\end{proof}
A PGAS program is not robust if and only if it has a violating computation $\tau$ of minimal length.
Let $\event\in\tau$ be the event determined by Lemma~\ref{Lemma:Cancellation}.
If $\kindOf{\event}\not\in\set{\readkind,\writekind}$, then $\tau=\tau_1\cdot\event\cdot\tau_2$.
Otherwise $\tau=\tau_1\cdot\event\cdot\tau_2\cdot\event'\cdot\tau_3\cdot\event''\cdot\tau_4$ with $\event\equivalenceorder\event'\equivalenceorder\event''$.
Consider the latter case where $\tau\setminus\event =\tau_1\cdot\tau_2\cdot\tau_3\cdot\tau_4$.
Since $\length{\tau\setminus\event}<\length{\tau}$, the new computation is not violating and $\happensbeforeOf{\tau\setminus\event}$ is acyclic.
This acyclicity guarantees we find a computation $\sigma\in\events^*$ with the same happens-before relation as $\tau\setminus\event$ and where pop events directly follow their remote accesses.
Intuitively, $\sigma$ is a sequentially consistent computation corresponding to $\tau\setminus\event$.
\begin{lemma}[\cite{ShashaSnir88}]
There is $\sigma\in\computationsOf{\program, \nodecount}$ with $\happensbeforeOf{\sigma}=\ \happensbeforeOf{\tau\setminus\event}$ and  $\sigma=\sigma_1\cdot \event_1\ldots\event_n\cdot\sigma_2$ for all $\event_1\equivalenceorder\ldots \equivalenceorder \event_n$.
\end{lemma}
We now use $\sigma$ to rearrange the events in $\tau\setminus\event$ and guarantee the normal-form requirement.
The idea is to project $\sigma$ to the events in $\tau_1$ to $\tau_4$.
Reinserting $\event$ yields a normal-form violation:
\begin{align*}
\tau\doubleprime:=(\projectionOf{\sigma}{\tau_1})\cdot\event\cdot (\projectionOf{\sigma}{\tau_2})\cdot\event'\cdot(\projectionOf{\sigma}{\tau_3})\cdot\event''\cdot(\projectionOf{\sigma}{\tau_4}).
\end{align*}
The following lemma concludes the proof of Theorem~\ref{Theorem:OnlyNeedWellShapedComputations}.
\begin{lemma}[Reinsertion]\label{Lemma:Reinsertion}
$\tau\doubleprime\in\computationsOf{\program, \nodecount}$, $\happensbeforeOf{\tau\doubleprime}=\ \happensbeforeOf{\tau}$, and $\tau\doubleprime$ is in normal form.
\end{lemma}
\proofatend
To relieve the reader from the burden of syntax, we consider the case when $\tau\setminus\event =\tau_1\cdot\tau_2$.
We start with the program order.
Let $\event_1,\event_2\in\tau_1$ with $\event_1\progorder\event_2$ in $\tau$ and, consequently, in $\tau\setminus\event$.
By definition of $\sigma$, we have $\event_1\progorder\event_2$ in $\sigma$.
Since $\projectionOf{\sigma}{\tau_1}$ contains $\event_1$ and $\event_2$ and does not add events between them, $\event_1\progorder\event_2$ holds for $\projectionOf{\sigma}{\tau_1}$ and, consequently, $\tau\doubleprime$.
Assume $\event_1\in\tau_1$ and $\event_2\in\tau_2$ with $\event_1\progorder\event_2$ in $\tau$ and in $\tau\setminus\event$.
Then $\event_1$ is the rightmost element in $\tau_1$ with its rank that is different from a pop.
Similarly, $\event_2$ is the leftmost element in $\tau_2$ with its rank and different from a pop.
The same is valid for their positions in $\projectionOf{\sigma}{\tau_1}$ and $\projectionOf{\sigma}{\tau_2}$, which leads to $\event_1\progorder\event_2$ in $\tau\doubleprime$.
The case when $\event_1\in\tau_1$ and $\event_2=\event$ is similar.
Since $\tau$ and $\tau\doubleprime$ consist of the same events, the cardinalities of the respective $\progorder$ relations are equal, and the above inclusion already means the program orders in both computations are equal.

Now we consider the conflict relation.
Let $\event_1,\event_2\in\tau_1$ with $\event_1\conflictorder\event_2$ in $\tau$ and hence in $\tau\setminus\event$.
By definition of $\sigma$, we have $\event_1\conflictorder\event_2$ in $\sigma$.
Since $\projectionOf{\sigma}{\tau_1}$ contains $\event_1$ and $\event_2$ and does not add new actions between them, $\event_1\conflictorder\event_2$ holds for $\projectionOf{\sigma}{\tau_1}$ and, consequently, for $\tau\doubleprime$.

Assume $\event_1,\event_2\in\tau_1$ and $\event_1\not\conflictorder\event_2$ in $\tau$.
One option is that $\event_1$ and $\event_2$ do not access the same address or both are reads.
Then they still will not conflict in $\tau\doubleprime$.
The other option is that $\event_1\conflictorder\event_3$ in $\tau$, where $\event_3$ is a write to $\addrOf{\event_1}=\addrOf{\event_2}$ that is located between $\event_1$ and $\event_2$ in $\tau_1$.
Then, as already proven, $\event_1\conflictorder\event_3$ will hold in $\tau\doubleprime$.
Consequently, $\event_1\conflictorder\event_2$ will not hold in $\tau\doubleprime$.
The case when $\event_1,\event_2\in\tau_2$ is similar.

Assume $\event_1\in\tau_1$, $\event_2\in\tau_2$, and $\event_1\conflictorder\event_2$ in $\tau$.
Then, $\event$ is not a write to $\addrOf{\event_1}=\addrOf{\event_2}$, and $\event_1\conflictorder\event_2$ in $\tau\setminus\event$.
Note that $\projectionOf{\sigma}{\tau_1}$ does not contain a write to $\addrOf{\event_1}$ to the right of $\event_1$.
Otherwise, $\tau_1$ would contain a write $\event_3$ to $\addrOf{\event_1}$, and $\event_1\conflictorder^+\event_3$, which contradicts $\event_1\conflictorder\event_2$ in $\tau$.
With a similar argument, $\projectionOf{\sigma}{\tau_2}$ does not contain a write to $\addrOf{\event_1}$ to the left of $\event_2$.
Therefore, $\event_1\conflictorder\event_2$ in $\projectionOf{\sigma}{\tau_1}\cdot\event\cdot\projectionOf{\sigma}{\tau_2}$.

Assume $\event_1\in\tau_1$, $\event_2\in\tau_2$, and $\event_1\not\conflictorder\event_2$ in $\tau$.
The proof of $\event_1\not\conflictorder\event_2$ in $\tau\doubleprime$ is as in the case when $\event_1,\event_2\in\tau_1$.

The case when $\event_1=\event$ or $\event_2=\event$ is no harder.

The formal definition of the identity relation takes a computation $\alpha$ and determines the three projections
$\projectionOf{\alpha}{\{\writekind, \readkind\}}$, $\projectionOf{\alpha}{\popakind}$, and $\projectionOf{\alpha}{\popbkind}$.
The identity relation then relates the $i$th elements in these projections.
To show that the identity relations in $\tau$ and $\tau\doubleprime$ coincide, one shows that the three projections coincide --- using the same technique as for the program order.
Therefore, the identity relations of both computations match.
Also note that for each read or write event sequence $\event_1\equivalenceorder\event_2\equivalenceorder\event_3$, we have $\event_1\before{\tau\doubleprime}\event_2\before{\tau\doubleprime}\event_3$.
This holds by the fact that $\event_1\before{\tau}\event_2\before{\tau}\event_3$, and the fact that $\sigma=\sigma_1\cdot \event_2\cdot \event_2\cdot\event_3\cdot \sigma_2$ for some $\sigma_1$ and $\sigma_2$.

To prove that $\tau\doubleprime\in\computationsOf{\program, \nodecount}$, we proceed by contradiction.
Let $\alpha\neq\tau\doubleprime$ be the longest prefix of $\tau\doubleprime$ so that $\initialeventstate\transitionto{\alpha}\state$ for some state $\state$.
Then $\tau\doubleprime=\alpha\cdot\tilde\event\cdot\beta$ with $\initialeventstate\transitionto{\alpha}\state$ and $\state\not\transitionto{\tilde\event}$.
Let $\state=(\stateconf,\memoryconf,\queueconfa,\queueconfb)$.
If $\kindOf{\tilde\event}\in\set{\popakind,\popbkind}$, then $\state\not\transitionto{\tilde\event}$ means that the respective queue $\queueconfa$ or $\queueconfb$ contains an incorrect topmost element or is empty in $\state$.
But this contradicts to $\event_1\before{\tau\doubleprime}\event_2\before{\tau\doubleprime}\event_3$ and equality of identity relations established above.
If $\kindOf{\tilde\event}\not\in\set{\popakind,\popbkind}$, then $\state\not\transitionto{\tilde\event}$ may hold because the transition
$\controlstate_1\transitionto{\command}\controlstate_2$ of $\tilde\event$ requires a different source state, $\controlstate_1\neq\stateconf(\rankOf{\tilde\event})$.
But since $\stateconf(\rankOf{\tilde\event})$ is unambiguously determined by the $\instructionOf{}$ of $\progorder$-predecessor of $\tilde\event$, which is the same in $\tau\doubleprime$ and in $\tau$ due to the matching program-order relations, this is not the case.
The last opportunity why $\state\not\transitionto{\tilde\event}$ may hold is because the transition producing $\tilde\event$ reads different values from registers or memory, e.g. $\tilde \event$ is an assertion $\theassert{\anexpr}$ and $\valueOf{\anexpr}=0$ in $\state$.
But since $\tau\doubleprime$ consists of the same events as $\tau$, has the same program and conflict relations (i.e. reads receive values from the same writes in both computations), and $\tau\in\computationsOf{\program}$, this cannot be the case.

Finally, $\tau\doubleprime$ is in normal-form. The condition on the shape of $\tau\doubleprime$ is immediate, \eqref{Equation:WellShaped} holds by the definitions of $\tau\doubleprime$ and $\sigma$.
\endproofatend
\begin{example}\label{Example:onetooneReshuffle}
Computation $\tau_{\onetoone}$ in Example~\ref{Example:onetooneComputation} is a shortest violation.
The event determined by Lemma~\ref{Lemma:Cancellation} is $\event=\nodeone{\loadkind}$.
Therefore, $\tau\setminus\event = \tau_1\cdot\tau_2$ with
\begin{align*}
\tau_1=\nodeone{\writekind}\cdot\nodetwo{\writekind}\cdot\nodetwo{\popakind}\cdot\nodeone{\popakind}\cdot\nodeone{\barrierkind}\cdot\nodetwo{\barrierkind}
\qquad \text{and}\qquad
\tau_2=\nodetwo{\popbkind}\cdot\nodeone{\popbkind}.
\end{align*}
A sequentially consistent computation corresponding to $\tau\setminus\event$ is
\begin{align*}
\sigma=\nodeone{\writekind}\cdot\nodeone{\popakind}\cdot\nodeone{\popbkind}\cdot\nodetwo{\writekind}\cdot\nodetwo{\popakind}\cdot\nodetwo{\popbkind}\cdot\nodeone{\barrierkind}\cdot\nodetwo{\barrierkind}.
\end{align*}
The normal-form violation $\tau_{\onetoone}\doubleprime$ is depicted in
Figure~\ref{Figure:onetooneCycleOnComputation}.
Note that $\tau_{\onetoone}\doubleprime$ is indeed in $\computationsOf{\onetoone, 2}$.
Moreover, $\nodeone{\popakind}$ and $\nodetwo{\popakind}$ immediately follow
$\nodeone{\writekind}$ and $\nodetwo{\writekind}$, respectively.
Similarly, the $\nodeone{\popbkind}$ and $\nodetwo{\popbkind}$ events in the second
part of the computation respect the order of $\nodeone{\writekind}$ and
$\nodetwo{\writekind}$ in the first part of the computation.
This means, \eqref{Equation:WellShaped} holds.
\end{example}
\begin{figure}
\centering 
	\begin{tikzpicture}[node distance=1.5cm, auto]
    \node (write0) {$\nodeone{\writekind}$}; 
    \node [minimum width=0.85cm,minimum height=0.35cm, right of=write0] (popa0) {};
    \node [node distance=0.06cm,below of=popa0] (popa0') {$\nodeone{\popakind}$}; 
    \node [right of=popa0] (write1) {$\nodetwo{\writekind}$}; 
    \node [minimum width=0.85cm,minimum height=0.35cm,right of=write1] (popa1) {};
    \node [node distance=0.06cm,below of=popa1] (popa1') {$\nodetwo{\popakind}$}; 
    \node [right of=popa1] (barrier0) {$\nodeone{\barrierkind}$}; 
    \node [right of=barrier0] (barrier1) {$\nodetwo{\barrierkind}$};
    \node [right of=barrier1] (load0) {$\nodeone{\loadkind}$}; 
    \node [minimum width=0.85cm,minimum height=0.35cm,right of=load0] (popb0) {};
    \node [node distance=0.06cm,below of=popb0] (popb0') {$\nodeone{\popbkind}$}; 
    \node [right of=popb0'] (popb1) {$\nodetwo{\popbkind}$};

  \draw[<->] (write0)	edge node [midway,above] {$\scriptstyle\equivalence$} (popa0);
  \draw[<->] (popa0)   edge[bend left=13] node [very near end,above] {$\scriptstyle\equivalence$} (popb0);
  \draw[->]  (write0)   	edge[bend right=15] node [midway,below] {$\scriptstyle po$} (barrier0);
  \draw[->]  (barrier0) 	edge[bend right=26] node [midway,below] {$\scriptstyle po$} (load0);
  \draw[<->] (barrier0)	edge node [midway,above] {$\scriptstyle\equivalence$} (barrier1);

  \draw[<->] (write1)  	edge node [midway,below] {$\scriptstyle\equivalence$} (popa1);
  \draw[<->] (popa1)  	edge[bend left=15] node [midway,above] {$\scriptstyle\equivalence$} (popb1);
  \draw[->]  (write1)  	edge[bend right=19] node [pos=0.50,below] {$\scriptstyle po$} (barrier1);

  \draw[->]  (load0)   	edge[bend right=22] node [midway,below] {$\scriptstyle\cf$} (popb1);

  \draw[-] (8.35,-0.75) edge (8.35,0.75);
  \draw[-] (9.65,-0.75) edge (9.65,0.75);
  \end{tikzpicture}
\caption{Normal-form violation $\tau\doubleprime_{\onetoone}$ from Example~\ref{Example:onetooneReshuffle}.
The edges indicate the dependencies in the computation and coincide with the relations in Figure~\ref{Figure:onetooneTrace}.}
\label{Figure:onetooneCycleOnComputation}
\end{figure}
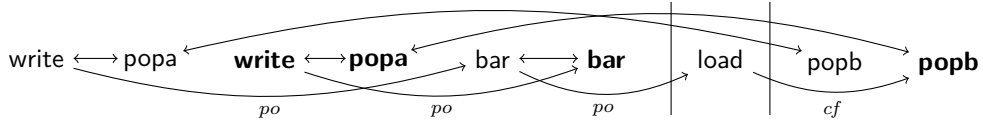

\section{From Normal-Form Violations to Language Emptiness}\label{Section:GeneratingViolatingWellShapedComputations}

We now reduce checking the absence of normal-form violations to the emptiness problem in
a suitable automaton model.
We introduce multiheaded automata and construct, for each program $(\program, \nodecount)$,
a multiheaded automaton accepting all normal-form computations.
To verify robustness, we check that the intersection of this automaton with regular
languages accepting cyclic happens-before relations is empty.

\subsection{Multiheaded Automata}\label{Subsection:MultiheadedAutomata}

Multiheaded automata are an extension of finite automata.
Intuitively, instead of generating just a single computation, they generate several computations in one pass, each by a separate head.
The language of the multiheaded automaton then consists of the concatenations of the computations generated by each head.

Syntactically, an \emph{$n$-headed finite automaton over alphabet $\Sigma$} is a finite automaton that uses the extended alphabet $\overline{1,n}\times\alphabet$.
So we have $\automaton=(\states, (\overline{1, n}\times \alphabet), \transitions, \initialstate, \finalstates)$.
The semantics, however, is different from finite automata.
Given $\sigma\in (\overline{1,n}\times\alphabet)^*$, we use $\projectionOf{\sigma}{k}$ to project $\sigma$ to the letters $(k, a)$, and afterwards cut away the index $k$.
So $\projectionOf{((1, a)\cdot(2, b)\cdot(1, c))}{1}=a\cdot c$.
With this, the \emph{language of $A$} is 
$	\langOf{\automaton}:=\setcond{\computOf{\sigma}}{\initialstate\transitionto{\sigma}\state\text{ for some }\state\in\finalstates}\quad\text{where}\quad \computOf{\sigma}:=\projectionOf{\sigma}{1}\cdots\projectionOf{\sigma}{n}$.

Multiheaded automata are closed under regular intersection, and
emptiness is decidable in non-deterministic logarithmic space.
Indeed, checking emptiness reduces to finding a path
from an initial to a final node in a directed graph.
\begin{lemma}\label{Lemma:MultiheadedIntersection}
Consider an $n$-headed automaton $U$ and a finite automaton $V$ over a common alphabet $\alphabet$.
There is an $n$-headed automaton $W$ with $\langOf{W}=\langOf{U}\cap\langOf{V}$.
\end{lemma}
\proofatend
Let $U=(\statesOf{U},\alphabet,\transitionsOf{U},\initialstateOf{U},\finalstatesOf{U})$
and $V=(\statesOf{V},\alphabet,\transitionsOf{V},\initialstateOf{V},\finalstatesOf{V})$.
We set $W:=(\statesOf{W},\alphabet,\transitionsOf{W},\initialstateOf{W},\finalstatesOf{W})$.
Let $\Omega$ be the set of functions $\overline{1,n}\to\statesOf{V}$.
Then, the set of states is $\statesOf{W}:=\set{\initialstateOf{W}}\uplus(\statesOf{U}\times\Omega\times\Omega)$.
The set of final states is $\finalstatesOf{W}:=\setcond{(\stateOf{U},\omega_1,\omega_2)}{\stateOf{U}\in\finalstatesOf{U},\text{ }\omega_1(n)\in\finalstatesOf{V},\text{ and }\omega_1(k)=\omega_2(k+1)\text{ for all }k\in\overline{1,n-1}}$.
The automaton has the following transitions:
\begin{itemize}
\item $\initialstateOf{W}\transitionto{\emptysequence}(\initialstateOf{U},\omega,\omega)$
for each $\omega\in\Omega$ with $\omega(1)=\initialstateOf{V}$,
\item $(\stateOf{U},\omega_1,\omega_2)\transitionto{k,a}(\stateOf{U}',\omega_1',\omega_2)$
if $\stateOf{U}\transitionto{k,a}\stateOf{U}'$, $\omega_1(k)\transitionto{a}\omega_1'(k)$, and $\omega_1(i)=\omega_1'(i)$
for $i\neq k$,
\item $(\stateOf{U},\omega_1,\omega_2)\transitionto{\emptysequence}(\stateOf{U}',\omega_1,\omega_2)$
if $\stateOf{U}\transitionto{\emptysequence}\stateOf{U}'$,
\item $(\stateOf{U},\omega_1,\omega_2)\transitionto{\epsilon}(\stateOf{U},\omega_1',\omega_2)$
if $\omega_1(k)\transitionto{\emptysequence}\omega_1'(k)$ and $\omega_1(i)=\omega_1'(i)$
for $i\neq k$.
\end{itemize}

Consider $\alpha=\alpha_1\cdots\alpha_n\in\langOf{U}\cap\langOf{V}$,
where $\alpha_k$ is produced by the $k^\text{th}$ head of $U$.
By the $\emptysequence$-transition from the initial state, $W$ guesses, for each $k$,
the state $\omega(k)$ that the automaton $V$ will reach after processing the prefix
$\alpha_1\cdots\alpha_{k-1}$ of $\alpha$.
The other transitions effectively execute the automaton $U$ synchronously with $n$
copies of the automaton $V$, each matching its own $\alpha_k$ subword of $\alpha$,
starting from the guessed initial state $\omega(k)$.
The set of final states $\finalstatesOf{W}$ makes sure that the guess was done correctly,
which means the $k^\text{th}$ copy of $V$ has reached the initial state of the $k+1^\text{th}$ copy, and the $n^\text{th}$
copy has reached a final state in $\finalstatesOf{V}$.
\endproofatend
\begin{lemma}\label{Lemma:EmptinessComplexity}
Emptiness for $n$-headed automata is \nl-complete.
\end{lemma}
Multiheaded automata are incomparable with context-free grammars, and indeed
the normal-form computations of a program may be non-context-free.\footnote{
	Consider $\program:=(\set{\controlstate_0},\commands,\set{\controlstate_0\transitionto{\theread{0}{0}{0}{0}}\controlstate_0},\set{\controlstate_0})$ running on a single node.
	The language $\computationsOf{\program,1}$ is not context-free.
	To see this, let $\kindOf{a}=\readkind$, $\kindOf{b}=\popakind$, and $\kindOf{c}=\popbkind$. Then $\computationsOf{\program,1}\cap a^* b^* c^*$ is the non-context-free
	language $\set{a^pb^pc^p\mid p\geq 0}$.
}
Multiheaded automata can be understood as a restriction of matrix grammars \cite{DP1989}.
In matrix grammars, productions simultaneously rewrite multiple non-terminals.
Roughly, each production can be understood as a Petri net transition, and emptiness is decidable as Petri net reachability is.
Since we target a \pspace\ result, matrix grammars are too expressive for our purposes.
\subsection{Detecting Normal-Form Computations}\label{Subsection:Generating}

We define a 4-headed automaton $\wsautomaton(\program,
\nodecount):=(\wsstates\uplus\statesOf{\wsautomaton}^{\aux},\events,
\wstransitions,\initialwsstate,\wsstates)$ that accepts all normal-form
computations
$\tau=\tau_1\cdot\tau_2\cdot\tau_3\cdot\tau_4\in\computationsOf{\program,
\nodecount}$.
In order to accept $\tau_1$, the new automaton tracks the control and
memory configurations in the way $\eventautomaton(\program, \nodecount)$ does.
For the remainder of the computation, these configurations are not needed.
Indeed, $\tau_2$ to $\tau_4$ only consist of $\popakind$ and
$\popbkind$ events that are executable independently of the control and memory
configurations.
However, $\wsautomaton(\program, \nodecount)$ has to take care of the ordering
of $\popakind$ and $\popbkind$ events from the same queue.
In particular, if $\event_1$ handles a request issued before the request of
$\event_2$ with $\kindOf{\event_1}=\kindOf{\event_2}$, then it cannot be the
case that $\event_1\in\tau_j$ and $\event_2\in\tau_i$ with $i < j$.

Guided by this discussion, we define a state $\state\in\wsstates$ as a
tuple $\state:=(\stateconf,\memoryconf,\partconfa, \partconfb)$.
The state and memory configurations $\stateconf$ and $\memoryconf$ are defined
as in Section~\ref{Section:Programs}.
They reflect the state of the program after it has generated a prefix of $\tau_1$.
The functions $\partconfa, \partconfb:\ranks\times\queuedomain\rightarrow
\parts$ give, for each process and each queue, the part $\tau_1$ to $\tau_4$
of the computation
where the next $\popakind$ resp. $\popbkind$ event will be generated.
The initial state is $\initialstate:=(\initialstateconf,\initialmemoryconf, \partconfa_0, \partconfb_0)$
with $\partconfa_0(\rank,\queueid):=1=:\partconfb_0(\rank,\queueid)$
for all $\rank\in\ranks$ and $\queueid\in\queuedomain$.

The transition relation $\wstransitions$ is the smallest relation defined by
the rules in the Tables~\ref{Table:WellShapedAutomatonRules} and~\ref{Table:WellShapedAutomatonRulesFull}.
Rule~($\guessa'$) lets the automaton choose the part of the computation to
which the next $\popakind$ event will be appended.
The first restriction is that the index of the part can only increase, as
events from the same queue are processed in order.
The second restriction is that $\popakind$ events cannot be generated to the
right of $\popbkind$ events from the same queue.
Rule~($\guessb'$) is the similar rule for $\popbkind$ events.

By Rule~($\writekind'$), the automaton appends a $\writekind$ event to $\tau_1$
and the corresponding $\popakind$ and $\popbkind$
events in one shot to the parts determined by $\partconfa$ and $\partconfb$.
Since a single transition of a multiheaded automaton can generate at most one
letter, the rule makes use of intermediary states
from $\statesOf{\wsautomaton}^{\aux}$.
If $\popbkind$ is added to $\tau_1$, the memory configuration is updated
accordingly.
Note that the generation in one shot causes pop events
within the same part $\tau_i$ to follow in the order of the corresponding
$\readkind$/$\writekind$ events in $\tau_1$.
Fortunately, this is always the case in normal-form
computations by~\eqref{Equation:WellShaped}.
Computations that are not in normal form, e.g.
$\tau_{\onetoone}$, cannot be generated by $\wsautomaton(\program, \nodecount)$.

The set of final states of $\wsautomaton(\program,\nodecount)$ is $\wsstates$.
The auxiliary states $\statesOf{\wsautomaton}^{\aux}$ are not included in the set of
final states to forbid computations with pending remote requests.

\begin{table}[t!]
\centering

\ifappendix
	\therules{
	
	
	\ifappendix
	
	  \therule
	  {$\command = \theload{\areg}{\anexpr}$}
	  {$\state\transitionto{1,(\loadkind,\rank,(\rank, \valueOf{\anexpr}\,))}\state'$\quad $\memoryconf':=\memoryconf[(\rank,\areg):=\memoryconf(\rank,\valueOf{\anexpr})]$}
	  {($\loadkind'$)}
	
	  \therule
	  {$\command = \thestore{\anexpr_{\anaddr}}{\anexpr_{\aval}}$}
	  {$\state\transitionto{1,(\storekind,\rank,(\rank,\anaddr))}\state'$\quad $\memoryconf':=\memoryconf[(\rank,\valueOf{\anexpr_{\anaddr}}):=\valueOf{\anexpr_{\aval}}]$}
	  {($\storekind'$)}
	
	  \therule
	  {$\command = \theassign{\areg}{\anexpr}$}
	  {$\state\transitionto{1,(\assignkind,\rank,\bottom)}\state'$\quad $\memoryconf':=\memoryconf[(\rank,\areg):=\valueOf{\anexpr}]$}
	  {($\assignkind'$)}
	
	  \therule
	  {$\command = \theassert{\anexpr}$\quad $\valueOf{\anexpr} \neq 0$}
	  {$\state\transitionto{1,(\assertkind,\rank,\bottom)}\state'$}
	  {($\assertkind'$)}
	\fi
	
	\ifappendix
	  \therule
	  {$\command=\theread{\anexpr_{\anaddr}^{\local}}{\anexpr_{\rank}^{\remote}}{\anexpr_{\anaddr}^{\remote}}{\anexpr_{\queueid}}$
	   \qquad
	   $\partconfa(\rank,\valueOf{\anexpr_{\queueid}})=m$
	   \qquad
	   $\partconfb(\rank,\valueOf{\anexpr_{\queueid}})=n$}
	  {$\state
	    \transitionto{1,(\readkind,\rank,\bottom)}
	    \state_{\textit{aux}1}
	    \transitionto{m,(\popakind,\rank,
	    (\valueOf{\anexpr_{\rank}^{\remote}},\valueOf{\anexpr_{\anaddr}^{\remote}}))}
	    \state_{\textit{aux}2}
	    \transitionto{n,(\popbkind,\rank,(\rank,\valueOf{\anexpr_{\anaddr}^{\local}}))}
	    \state'$\quad $\stateconf':=\stateconf[\rank:=\controlstate_2]$\\
	     if $n=1$ then $\memoryconf':=\memoryconf[(\rank,\valueOf{\anexpr_{\anaddr}^{\local}}):= \memoryconf(\valueOf{\anexpr_{\rank}^{\remote}},\valueOf{\anexpr_{\anaddr}^{\remote})}]$
	  }
	  {($\readkind'$)}
	\fi
	
	
	\ifappendix
	  \therule
	  {$\stateconf(\rank)\transitionto{\thebarrier}\stateconf'(\rank)$ for each $\rank\in\ranks$}
	  {$\state
	    \transitionto{1,(\barrierkind,1,\bottom)}
	    \state_{\textit{aux}1}
	    \transitionto{1,(\barrierkind,2,\bottom)}
	    \ldots
	    \transitionto{1,(\barrierkind,\nodecount,\bottom)}
	    (\stateconf', \memoryconf, \partconfa, \partconfb)$}
	  {($\barrierkind'$)}
	\fi
	}
\else
	\begin{tabular}{cc}
	\hspace{-1.5em}
	\therules{
	  \theruleleft
	  {$\partconfa(\rank,\queueid)<\partconfb(\rank, \queueid)$}
	  {$\state\transitionto{\varepsilon}\state'$~$\partconfa':=\partconfa[(\rank,\queueid):=\partconfa(\rank,\queueid)+1]$}
	  {($\guessa'$)}
	}
	& \hspace{-1em}
	\therules{
	
	  \therule
	  {$\partconfb(\rank, \queueid) < 4$}
	  {$\state\transitionto{\varepsilon}\state'$~$\partconfb':=\partconfb[(\rank,\queueid):=\partconfb(\rank,\queueid)+1]$}
	  {($\guessb'$)}
	
	}
	\\[0.6cm]
	\multicolumn{2}{c}{
	\therules{
	  \therule
	  {$\command=\thewrite{\anexpr_{\anaddr}^{\local}}{\anexpr_{\rank}^{\remote}}{\anexpr_{\anaddr}^{\remote}}{\anexpr_{\queueid}}$
	   \qquad
	   $\partconfa(\rank,\valueOf{\anexpr_{\queueid}})=m$
	   \qquad
	   $\partconfb(\rank,\valueOf{\anexpr_{\queueid}})=n$}
	  {$\state
	    \transitionto{1,(\writekind,\rank,\bottom)}
	    \state_{\textit{aux}1}
	 \transitionto{m,(\popakind,\rank,(\rank,\valueOf{\anexpr_{\anaddr}^{\local}}))}
	    \state_{\textit{aux}2}
	 \transitionto{n,(\popbkind,\rank,(\valueOf{\anexpr_{\rank}^{\remote}},\valueOf{
	\anexpr_{\anaddr}^{\remote}}))}
	    \state'$\quad $\stateconf':=\stateconf[\rank:=\controlstate_2]$\\
	if $n=1$ then $\memoryconf':=\memoryconf[(
	\valueOf{\anexpr_{\rank}}^{\remote},\valueOf{\anexpr_{\anaddr}^{\remote}}
	):=\memoryconf(\rank,\valueOf{\anexpr_{\anaddr}^{\local})}]$
	  }
	  {($\writekind'$)}
	}}
	\end{tabular}
\fi

\caption{Transition rules for $\wsautomaton(\program,\nodecount)$, given 
$\controlstate_1\transitionto{\command}\controlstate_2$ and current state 
$\state=(\stateconf,\memoryconf,\partconfa, \partconfb)$ with   
$\stateconf(\rank)=\controlstate_1$. 
The target is $\state'=(\stateconf',\memoryconf',\partconfa', \partconfb')$ where,
unless otherwise stated, $\stateconf'=\stateconf$, $\memoryconf'=\memoryconf$, 
$\partconfa'=\partconfa$, $\partconfb'=\partconfb$.
The auxiliary states $\state_{\mathit{aux}1}, \state_{\mathit{aux}2}\in \statesOf{\wsautomaton}^\aux$ 
are unique for each rule application.
}
\label{Table:WellShapedAutomatonRules\ifappendix Full\fi}
\end{table}

\lemmaatend
\begin{lemma}\label{Lemma:WSAutomatonGeneratesValidComputations}
$\wsautomaton(\program, \nodecount)$ only generates computations of $(\program, \nodecount)$: $\langOf{\wsautomaton(\program, \nodecount)}\subseteq\computationsOf{\program, \nodecount}$.
\end{lemma}
\begin{proof}
Consider $\initialwsstate\transitionto{\sigma}\stateOf{\wsautomaton}$ with $\stateOf{\wsautomaton}=(\stateconf,\memoryconf,\partconfa,\partconfb)\in\wsstates$.
Let $\tau=\computOf{\sigma}=\tau_1\cdot\tau_2\cdot\tau_3\cdot\tau_4$ with $\tau_i=\projectionOf{\sigma}{i}$.
We prove the following by induction on the length of the computation.

\begin{description}
\item[IS1] $\initialeventstate\transitionto{\tau}\stateOf{\eventautomaton}$ for some $\stateOf{\eventautomaton}\in\finaleventstates$. Membership in $\finaleventstates$ means the queues of $\stateOf{\eventautomaton}$ are empty.
\item[IS2] $\initialeventstate\transitionto{\tau_1}(\stateconf,\memoryconf,\queueconfa,\queueconfb)$ for some $\queueconfa$, $\queueconfb$, but with the same $\stateconf$, $\memoryconf$ as in $\stateOf{\wsautomaton}$ above.
\item[IS3] Let $\partconfa(\rank, \queueid)=k$. Then no $\tau_i$ with $i>k$ contains an event $\event$ with $\kindOf{\event}=\popakind$, $\rankOf{\event}=\rank$, and $\queueOf{\event}=\queueid$. A similar statement holds for $\queueconfb$.
\item[IS4] For all $\event\in\tau_2\cdot\tau_3\cdot\tau_4$ we have $\kindOf{\event}\in\set{\popakind,\popbkind}$.
\end{description}

In the base case with $\sigma=\emptysequence$ the inductive statement trivially holds.

Assume the statement holds for $\sigma$.
Consider $\initialwsstate\transitionto{\sigma'}\stateOf{\wsautomaton}'=(\stateconf',\memoryconf',\partconfa', \partconfb')$ which extends $\sigma$ with Rule~($\readkind'$):
\begin{align*}
\sigma'=\sigma\cdot(1,\event_1)\cdot(2,\event_2)\cdot(3,\event_3)\qquad  \kindOf{\event_1}=\readkind,\ \kindOf{\event_2}=\popakind,\ \kindOf{\event_3}=\popbkind.
 \end{align*}
Then $\memoryconf'=\memoryconf$, $\partconfa'=\partconfa$, $\partconfb'=\partconfb$, and $\tau'=\computOf{\sigma'}=\tau_1'\cdot\tau_2'\cdot\tau_3'\cdot\tau_4'$, where $\tau_i'=\projectionOf{\sigma'}{i}$ are $\tau_1'=\tau_1\cdot\event_1$, $\tau_2'=\tau_2\cdot\event_2$, $\tau_3'=\tau_3\cdot\event_3$, and $\tau_4'=\tau_4$.
Since \textbf{IS4} and \textbf{IS3} hold for $\sigma$, they also hold for $\sigma'$ by definition of $\sigma'$ and Rule~($\readkind'$).

It remains to check the behaviour of the state-space automaton.
By \textbf{IS2} from the induction hypothesis and the Rules~(\readkind) and~(\readkind'), we have $\initialeventstate\transitionto{\tau_1\cdot\event_1}(\stateconf',\memoryconf,\queueconfa',\queueconfb)$.
So \textbf{IS2} holds for $\sigma'$ as well.
To check \textbf{IS1} for $\sigma'$, we consider the content of $\queueconfa'$.
According to Rule~($\readkind$), we have $\queueconfa':=\queueconfa[(\rankOf{\event_1},\queueOf{\event_1}):=\queueconfa(\rankOf{\event_1},\queueOf{\event_1})\cdot(\rank_{\remote},\anaddr_{\remote},\rank_{\local},\anaddr_{\local})]$.
By the induction hypothesis, we can generate $\tau_2$ from $(\stateconf,\memoryconf,\queueconfa,\queueconfb)$.
In $(\stateconf',\memoryconf,\queueconfa',\queueconfb)$, we append an action to $\queueconfa$.
Since $\tau_2$ only consists of $\popakind$ and $\popbkind$ events,
we can still generate the computation from $(\stateconf',\memoryconf,\queueconfa',\queueconfb)$. This yields $\initialeventstate\transitionto{\tau_1\cdot\event_1\cdot\tau_2}\state_1$ for some $\state_1$.

We now show that $\state_1\transitionto{\event_2}\state_2$ for some $\state_2$.
Let $\state_1=(\stateconf'',\memoryconf'',\queueconfa'',\queueconfb'')$.
When checking \textbf{IS3} for $\sigma'$, we noted that $\tau_3\cdot\tau_4$ does not contain $\popakind$ events $\tilde\event$ with rank $\rankOf{\tilde\event}=\rankOf{\event_1}$ and queue id $\queueOf{\tilde\event}=\queueOf{\event_1}$.
Therefore, by \textbf{IS1} from the induction hypothesis, all elements in $\queueconfa(\rankOf{\event_1},\queueOf{\event_1})$ are popped by $\popakind$ transitions in $\tau_2$.
As a result, $\queueconfa''(\rankOf{\event_1},\queueOf{\event_1})$ contains only the single element added by $\event_1$.
Comparing Rules~($\readkind$),~($\popakind$), and~($\readkind$'), shows $\state_1\transitionto{\event_2}\state_2$.
Note that we need to take the read-rules into account to make sure the contents of the tuple $\event_2$ coincide for $\wsautomaton(\program, \nodecount)$ and $\eventautomaton(\program, \nodecount)$.

The fact that $\eventautomaton(\program, \nodecount)$ can accept the rest of computation $\tau'$ ($\state_2\transitionto{\tau_3\cdot\event_3\cdot\tau_4}\state_3$ for some $\state_3$) is proven similarly.
Emptiness of the queues in $\state_3$ follows from Rule~($\readkind$') and \textbf{IS1} for $\tau$.

The argumentation for write events, $\kindOf{\event_1}=\writekind$, is the same.
For the remaining kinds of events $\event_1$, the proofs are simpler.
There, we only need to make use of state and memory configurations, which coincide in $\wsautomaton(\program, \nodecount)$ and $\eventautomaton(\program, \nodecount)$.
\end{proof}

\begin{lemma}\label{Lemma:WSAutomatonGeneratesAllWellShapedComputations}
Automaton $\wsautomaton(\program, \nodecount)$ generates all normal-form computations of the program: $\setcond{\tau\in\computationsOf{\program, \nodecount}}{\tau\text{ is in normal form}}\subseteq\langOf{\wsautomaton(\program, \nodecount)}$.
\end{lemma}
\begin{proof}
Consider a normal-form computation $\tau=\tau_1\cdot\tau_2\cdot\tau_3\cdot\tau_4\in\computationsOf{\program, \nodecount}$ with $\initialeventstate\transitionto{\tau_1}\stateOf{\eventautomaton}$ for some $\stateOf{\eventautomaton}=(\stateconf,\memoryconf,\queueconfa,\queueconfb)$.
To prove that $\wsautomaton(\program, \nodecount)$ can generate $\tau$, we show the following by induction on the length of the computation. (Note that by~\eqref{Equation:WellShaped} we can extend normal-form computations inductively).

\begin{description}
\item[IS1] $\initialwsstate\transitionto{\sigma}\stateOf{\wsautomaton}=(\stateconf,\memoryconf,\partconfa, \partconfb)$ with $\stateconf$ and $\memoryconf$ from $\stateOf{\eventautomaton}$ above.
\item[IS2] We have $\projectionOf{\sigma}{i}=\tau_i$ for all $i\in\overline{1, 4}$.
\item[IS3] Let the last $\event$ with $\kindOf{\event}=\popakind$, $\rankOf{\event}=\rank$, $\queueOf{\event}=\queueid$ be in $\tau_k$. Then $\partconfa(\rank, \queueid)=k$. If there is no such event, $\partconfa(\rank, \queueid)=1$. There is a similar requirement for $\popbkind$ events.
\end{description}

Note that computation $\emptysequence$ satisfies all the constraints.
Assume the constraints hold for computation $\tau$.
We extend $\tau$ to a computation $\tau'=\tau_1'\cdot\tau_2'\cdot\tau_3'\cdot\tau_4'$, and show that it also satisfies \textbf{IS1} to \textbf{IS3}.
Extending $\tau$ adds an event to the first part of the computation, $\stateOf{\eventautomaton}\transitionto{\event_1}\stateOf{\eventautomaton}'$.
We do a case distinction based on $\kindOf{\event_1}$.

Consider the case $\kindOf{\event_1}=\readkind$.
Let $\event_1\equivalenceorder\event_2\equivalenceorder\event_3$ with $\tau_2'=\tau_2\cdot \event_2$ and $\tau_3'=\tau_3\cdot \event_3$.
Assume $\event_1$ was generated by the transition $\controlstate_1\transitionto{\command}\controlstate_2$.
This means $\stateconf(\rankOf{\event_1})=\controlstate_1$.
By \textbf{IS1} in the induction hypothesis, $\stateOf{\eventautomaton}$ and $\stateOf{\wsautomaton}$ share the same $\stateconf$ and $\memoryconf$.
Therefore, by Rules~(\readkind) and~(\readkind'), $\wsautomaton(\program, \nodecount)$ can mimic the read in $\eventautomaton(\program, \nodecount)$.
To make sure we append $\event_2$ to $\tau_2$, we have to check the requirements on $\partconfa$.
If $\partconfa(\rankOf{\event_2}, \queueOf{\event_2})< 2$, we can use Rule~(\guessa') to adapt the counter.
If we assume that $\partconfa(\rankOf{\event_2}, \queueOf{\event_2})= k > 2$, we derive a contradiction as follows.
By the induction hypothesis, there is an event $\event'$ in $\tau_k$ with
$\rankOf{\event'}=\rankOf{\event_2}$, $\queueOf{\event'}=\queueOf{\event_2}$, and $\kindOf{\event'}=\kindOf{\event_2}=\popakind$.
This event has a corresponding event $\event\equivalenceorder \event'$ in $\tau_1$.
Summing up, $\event$, $\event_1$, $\event_2$, $\event'$ are contained in $\tau$ in this order. Moreover, the latter two events are added to the same queue in reverse order: $\event'$ before $\event_2$. A contradiction to the definition of FIFO.
We conclude
\begin{align*}
\stateOf{\wsautomaton}\transitionto{(1,\event_1)\cdot(2,\event_2)\cdot(3,\event_3)}\stateOf{\wsautomaton}'.
\end{align*}
The requirements \textbf{IS1} to \textbf{IS3} are readily checked.
The argumentation for write events is the same.
For the remaining kinds of events, the induction step is simpler since $\stateconf$ and $\memoryconf$ coincide in $\stateOf{\eventautomaton}$ and $\stateOf{\wsautomaton}$.
\end{proof}
\endlemmaatend

\begin{lemma}\label{Lemma:WSAutomatonWorks}
$\setcond{\tau\in\computationsOf{\program, \nodecount}}{\tau\text{ is in normal form}}=\langOf{\wsautomaton(\program, \nodecount)}$.
\end{lemma}
\proofatend
The inclusion from left to right is Lemma~\ref{Lemma:WSAutomatonGeneratesAllWellShapedComputations}.
The inclusion from right to left holds by Lemma~\ref{Lemma:WSAutomatonGeneratesValidComputations} and the observation that $\wsautomaton(\program, \nodecount)$ only generates computations in normal form. \endproofatend

\lemmaatend
The following lemma states that $\wsautomaton(\program, \nodecount)$ generates events in program order.
\begin{lemma}\label{Lemma:GeneratesInProgramOrder}
Consider computation $\initialwsstate\transitionto{\sigma}\stateOf{\wsautomaton}$ with events $(1,\event_1)\before{\sigma}(1,\event_2)$ so that $\kindOf{\event_1}, \kindOf{\event_2}\notin \{\popakind, \popbkind\}$ and $\rankOf{\event_1}=\rankOf{\event_2}$.
Then $\event_1\progorder^+\event_2$ in $\tau=\computOf{\sigma}$.
\end{lemma}
\begin{proof}
By definition of the transition relation $\wstransitions$ and $\progorder$.
\end{proof}
\endlemmaatend

\lemmaatend
The following lemma states that $\wsautomaton(\program, \nodecount)$ generates the events $\popakind$ and $\popbkind$ immediately after the corresponding $\readkind$ or $\writekind$ event.
\begin{lemma}\label{Lemma:GeneratesPopsImmediately}
Let $\initialwsstate\transitionto{\sigma}\stateOf{\wsautomaton}$,  $\tau=\computOf{\sigma}$, and $\event_1, \event_2, \event_3\in\tau$ with $\kindOf{\event_1}\in\{\readkind, \writekind\}$, $\kindOf{\event_2}=\popakind$, and $\kindOf{\event_3}=\popbkind$. Then $\event_1\equivalenceorder\event_2\equivalenceorder\event_3$ holds in $\tau$ if and only if $\sigma=\sigma_1\cdot(1,\event_1)\cdot(m,\event_2)\cdot(n,\event_3)\cdot\sigma_2$ for some $\sigma_1$, $\sigma_2$ and $m,n\in\parts$
with $m\leq{}n$.
\end{lemma}
\begin{proof}
By Rules ($\readkind$) and ($\writekind$), the preconditions on ($\guessa$) and ($\guessb$), and the definition of $\equivalenceorder$.
\end{proof}
\endlemmaatend

\subsection{Detecting Violations}\label{Subsection:DetectingHappensBeforeCycles}

The multiheaded automaton accepts all normal form computations, and we would like to check
if one of those computations is violating.
In general, violating computations can contain complicated cycles in the happens-before
relation.
However, we now show that whenever a computation has a happens-before cycle, it
has a cycle in which each process is entered and left at most once.
Our algorithm for robustness will look for happens-before cycles of this special form that,
as we will show, can be captured by a regular language.

\begin{lemma}\label{Lemma:EachRankContributesOnce}
Computation $\tau\in\computationsOf{\program, \nodecount}$ is violating iff there is a cycle
\begin{multline}
a_1\equivalenceorder^*b_1\progorder^*c_1\equivalenceorder^*d_1\cycorder\dots\cycorder{}
a_k\equivalenceorder^*b_k\progorder^*c_k\equivalenceorder^*d_k\cycorder{}a_1\qquad
\tag{CYC}\label{Tag:CyclicChain}
\end{multline}
where $\rankOf{x_i}= \rankOf{y_j}$ iff $i= j$, for all $x_i, y_j\in\{a_1,\ldots, d_k\}$, and $\cycorder\ :=\ \conflictorder\cup\equivalenceorder$.
\end{lemma}
\proofatend
Consider an arbitrary cycle.
It has the following form:
\begin{align*} a_1\equivalenceorder^* b_1\progorder^*c_1\equivalenceorder^* d_1\cycorder\ldots\cycorder a_n\equivalenceorder^*b_n\progorder^*c_n\equivalenceorder^* d_n\cycorder a_1.
\end{align*}
Assume now $\rankOf{a_i}=\ldots=\rankOf{d_i}=\rankOf{a_j}=\ldots=\rankOf{d_j}$ for some $i<j$.
Fix these $i$ and $j$.
Then either $b_i\progorder^*c_j$ or $b_j\progorder^*c_i$.
In the former case, $\tau$ has the following happens-before cycle:
\begin{align*}
&a_1\equivalenceorder^* b_1\progorder^*c_1\equivalenceorder^* d_1\cycorder\ldots\cycorder a_i \equivalenceorder^*b_i\progorder^*c_j\equivalenceorder^* d_j\cycorder\ldots\\
&\hspace{7cm}
\cycorder a_n\equivalenceorder^*b_n\progorder^*c_n\equivalenceorder^* d_n\cycorder a_1.
\end{align*}
In the latter case, $\tau$ has the following cycle:
\begin{align*}
a_j\equivalenceorder^*b_j\progorder^*c_i\equivalenceorder^* d_i\cycorder\ldots\cycorder a_{j-1}\equivalenceorder^*b_{j-1}\progorder^*c_{j-1}\equivalenceorder^*d_{j-1}\cycorder a_j.
\end{align*}
Repeating the procedure for the new cycle until there is no $i\neq j$ with $\rankOf{a_i}=\ldots=\rankOf{d_i}=\rankOf{a_j}=\ldots=\rankOf{d_j}$, we get a cycle of the desired form.
\endproofatend

\begin{example}
The computations $\tau_{\onetoone}$ (Example~\ref{Example:onetooneComputation})
and $\tau\doubleprime_{\onetoone}$ (Example~\ref{Example:onetooneReshuffle}) have
a cycle of the form~\eqref{Tag:CyclicChain} depicted in Figure~\ref{Figure:onetooneTrace}: $n=2$, $a_1=b_1=\nodeone{\barrierkind}$, $c_1=d_1=\nodeone{\loadkind}$,  $a_2=\nodetwo{\popbkind}$, $b_2=\nodetwo{\writekind}$, $c_2=d_2=\nodetwo{\barrierkind}$.
\end{example}

Note that $d_i\equivalenceorder a_{i+1}$ means both are barriers, $\kindOf{d_i}=\barrierkind=\kindOf{a_{i+1}}$. This holds as the ranks are different.
In spite of the additional restrictions, cycles~\eqref{Tag:CyclicChain} are not trivial to recognize.
The reason is that the events constituting the cycle are not necessarily contained in the computation in the order in which they appear in the cycle, see Figure~\ref{Figure:onetooneCycleOnComputation}.
The idea of our cycle detection is to first guess the events $a_i$ and $d_i$ for each process and then check that $d_i\conflictorder{}a_{i+1}$ holds.
The former can be accomplished by an extension $\wsautomatonprime(\program, \nodecount)$ of the multiheaded automaton $\wsautomaton(\program, \nodecount)$, the latter by a regular intersection.

The automaton $\wsautomatonprime(\program, \nodecount)$ accepts computations over the alphabet $\events\times\markings$ with $\markings:=\powerset{\set{\enterguess,\leaveguess}}$.
The events marked by $\enterguess$ are the guessed $a_i$ events in \eqref{Tag:CyclicChain} and
those marked by $\leaveguess$ are the $d_i$ events in \eqref{Tag:CyclicChain}.
We still have to guarantee we only mark $a_i$ and $d_i$ that satisfy $a_i\equivalenceorder^*b_i\progorder^*c_i\equivalenceorder^*d_i$.
This is straightforward thanks to the fact that $\wsautomaton(\program, \nodecount)$ generates
the events of each process in program order, and generates events related by $\equivalenceorder$ in one shot.
The full construction of $\wsautomatonprime(\program, \nodecount)$ is given in the appendix.

\lemmaatend
Now we formally define the automaton $\wsautomatonprime(\program, \nodecount)$, which is an extension of $\wsautomaton(\program, \nodecount)$ that non-deterministically guesses and marks the first and the last event in each process that contribute to a cycle --- if any.
We set $\wsautomatonprime(\program, \nodecount):=(\wsstatesprime,\events\times\markings,\wstransitionsprime,\initialwsstateprime,\finalwsstatesprime)$,
where events are optionally marked by $\enterguess$ and/or $\leaveguess$ from $\markings:=\powerset{\set{\enterguess,\leaveguess}}$.
The events marked by $\enterguess$ are the $a_i$ events in \eqref{Tag:CyclicChain}
and those marked by $\leaveguess$ are the $d_i$ events in \eqref{Tag:CyclicChain}.
The set of states $\wsstatesprime$ consists of the states $\wsstates$ extended by information about which marked events have been issued for each process: $\wsstatesprime:=\wsstates\times\set{\bottom,\enterguess,\leaveguess}^{\ranks}$.
The initial state is $\initialwsstateprime:=(\initialwsstate,\initialmarkingfunc)$ with $\initialmarkingfunc(\rank):=\ \bottom$ for each rank.
The transition relation $\wstransitionsprime$ is defined as follows:

\begin{description}
\item[M1] $(\state,\markingfunc)\transitionto{\emptysequence}(\state',\markingfunc)$ if $\state\transitionto{\emptysequence}\state'$.
\item[M2] $(\state,\markingfunc)\transitionto{i,(\event,\emptyset)}(\state',\markingfunc)$ if $\state\transitionto{i,\event}\state'$.
\item[M3] $(\state,\markingfunc)\transitionto{i,(\event,\set{\enterguess})}(\state',\markingfunc[\rankOf{\event}:=\enterguess])$ if $\state\transitionto{i,\event}\state'$, $\addrOf{\event}\neq\ \bottom$ or $\kindOf{\event}=\barrierkind$, and $\markingfunc(\rankOf{\event})=\ \bottom$.
\item[M4] $(\state,\markingfunc)\transitionto{i,(\event,\set{\enterguess,\leaveguess})}(\state',\markingfunc[\rankOf{\event}:=\leaveguess])$ if $\state\transitionto{i,\event}\state'$, $\addrOf{\event}\neq\ \bottom$ or $\kindOf{\event}=\barrierkind$, and $\markingfunc(\rankOf{\event})=\ \bottom$.
\item[M5] $(\state,\markingfunc)\transitionto{i,(\event,\set{\leaveguess})}(\state,\markingfunc[\rankOf{\event}:=\leaveguess])$ if $\state\transitionto{i,\event}\state'$, $\addrOf{\event}\neq\ \bottom$ or $\kindOf{\event}=\barrierkind$, and $\markingfunc(\rankOf{\event})=\enterguess$.
\item[M6] $(\state,\markingfunc)\transitionto{i,(\event_1,\set{\leaveguess})}\state_{\textit{aux}}\transitionto{j,(\event_2,\set{\enterguess})}(\state',\markingfunc[\rankOf{\event_1}:=\leaveguess])$ if $\state\transitionto{i,\event_1}\state_{\textit{aux}}\transitionto{j,\event_2}\state'$, $\kindOf{\event_1}=\popakind$, $\kindOf{\event_2}=\popbkind$, and $\markingfunc(\rankOf{\event_1})=\ \bottom$.
\end{description}

The set of final states is $\finalwsstatesprime:=\setcond{(\state,\markingfunc)}{\state\in\finalwsstates\text{ and }\markingfunc(\rank)\in\set{\bottom,\leaveguess}\text{ for all }\rank\in\ranks}$.

\begin{lemma}\label{Lemma:MarkedWSAutomatonWorksAsBefore}
The languages of $\wsautomaton(\program, \nodecount)$ and $\wsautomatonprime(\program, \nodecount)$ match up to the markings: $\langOf{\wsautomaton(\program, \nodecount)}=\projectionOf{\langOf{\wsautomatonprime(\program, \nodecount)}}{\events}$.
\end{lemma}
\begin{proof}
The inclusion $\langOf{\wsautomaton(\program, \nodecount)}\subseteq\projectionOf{\langOf{\wsautomatonprime(\program, \nodecount)}}{\events}$ holds due to the Rules~\textbf{M1} and~\textbf{M2} in the definition of $\wstransitionsprime$.
The reverse inclusion $\langOf{\wsautomaton(\program, \nodecount)}\supseteq\projectionOf{\langOf{\wsautomatonprime(\program, \nodecount)}}{\events}$ follows from the fact that $(\state,\markingfunc)\transitionto{i,\event,\marking}(\state',\markingfunc')$ requires $\state\transitionto{i,\event}\state'$ (\textbf{M2}-\textbf{M6}).
\end{proof}

\begin{lemma}\label{Lemma:MarkedWSAutomatonGuesses}
Consider a marked computation $\tau\in\langOf{\wsautomatonprime(\program, \nodecount)}$ and events $(a,\marking_1)$ and $(d,\marking_4)$ in $\tau$ with $\rankOf{a}=\rankOf{d}$ and $\enterguess\in\marking_1$, $\leaveguess\in\marking_4$.
Then $a\equivalenceorder^*b\progorder^*c\equivalenceorder^*d$ for some $(b,\marking_2)$ and $(c,\marking_3)$ in $\tau$.
\end{lemma}
\begin{proof}
Consider $\initialwsstateprime\transitionto{\sigma}\stateOf{\wsautomatonprime}$ and let $\tau=\computOf{\sigma}$.
Let $(a,\marking_1)$ and $(d,\marking_4)$ be two events in $\tau$ with $\rankOf{a}=\rankOf{d}$ and $\enterguess\in\marking_1$, $\leaveguess\in\marking_4$.
Then, $\sigma$ contains $(i,a,\marking_1)$ and $(j,d,\marking_4)$ for some $i,j\in\parts$.
\begin{itemize}
\item
If $(i,a,\marking_1)\after{\sigma}(j,d,\marking_2)$, then $a$ and $d$ were generated by the two transitions defined by Rule~\textbf{M6}.
This means $\sigma=\sigma_1\cdot(1,b,\emptyset)\cdot(j,d,\set{\leaveguess})\cdot(i,a,\set{\enterguess})\cdot\sigma_2$, where $\kindOf{b}\in\set{\readkind,\writekind}$, $\kindOf{d}=\popakind$, $\kindOf{a}=\popbkind$.
Therefore, $b\equivalenceorder d\equivalenceorder a$, which can be reformulated as $a\equivalenceorder^*b\progorder^*b\equivalenceorder^*d$.

\item
If $(i,a,\marking_1)=(j,d,\marking_4)$, then $\marking_1=\marking_4=\set{\enterguess,\leaveguess}$ and $a=d$ is the event generated by~\textbf{M4}. Clearly, $a\equivalenceorder^*a\progorder^*d\equivalenceorder^*d$.

\item
If $(i,a,\marking_1)\before{\sigma}(j,d,\marking_4)$, then $a$ was generated by \textbf{M3}, and $d$ was generated by Rule~\textbf{M5}.
Let $b\equivalenceorder^*a$ with $\kindOf{b}\notin\{\popakind, \popbkind\}$, and similarly $c\equivalenceorder^* d$.
The fact that $b\progorder c$ follows from Lemma~\ref{Lemma:GeneratesPopsImmediately} and Lemma~\ref{Lemma:GeneratesInProgramOrder}.
Altogether, $a\equivalenceorder^*b\progorder^*c\equivalenceorder^*d$.
\end{itemize}
\end{proof}

For the next lemma, consider a normal-form computation $\tau\in\computationsOf{\program, \nodecount}$ and let $\set{\rank_1\ldots\rank_n}\subseteq\ranks$ be a set of ranks.
Moreover, assume that for each rank $\rank_i$ with $i\in \overline{1, n}$, there are $a_i, b_i, c_i, d_i\in \tau$ that have this rank, satisfy $a_i\equivalenceorder^*b_i\progorder^*c_i\equivalenceorder^*d_i$, and where
\begin{align*}
(\addrOf{a_i}\neq\ \bottom\text{ or }\kindOf{a_i}=\barrierkind)\text{ and }(\addrOf{d_i}\neq\ \bottom\text{ or }\kindOf{d_i}=\barrierkind).
\end{align*}
\begin{lemma}\label{Lemma:MarkedWSAutomatonGuessesEverything}
Under these assumptions, there is a marked computation $\tau'\in\langOf{\wsautomatonprime(\program, \nodecount)}$ with $\projectionOf{\tau'}{\events}=\tau$ that contains, for each $i\in\overline{1,n}$, a marked event $(a_i,\marking_1)$ with $\enterguess\in\marking_1$ and $(d_i,\marking_4)$ with $\leaveguess\in\marking_4$.
All other marked events $(\event,\marking)\in\tau$ have $\marking=\emptyset$.
\end{lemma}
\begin{proof}
We prove the statement of the lemma by induction on the size $n$ of the set of ranks.
The base case $n=0$ is due to Lemma~\ref{Lemma:WSAutomatonGeneratesAllWellShapedComputations} and the Rules~\textbf{M1} and~\textbf{M2}: $\wsautomatonprime(\program, \nodecount)$ can generate a marked computation $\tau_0$ with $\projectionOf{\tau_0}{\events}=\tau$ and all markings being $\emptyset$.
Formally, there is $\initialwsstateprime\transitionto{\sigma_0}\stateOf{\wsautomatonprime}$ for some $\stateOf{\wsautomatonprime}\in\finalwsstatesprime$ with $\tau_0=\computOf{\sigma_0}$.

In the induction step, assume the claim holds for sets of ranks of size $n-1$ and consider $\{\rank_1,\ldots, \rank_n\}\subseteq \ranks$. By the hypothesis, there is $\initialwsstateprime\transitionto{\sigma_{n-1}}\stateOf{\wsautomatonprime}$ for some $\stateOf{\wsautomatonprime}\in\finalwsstatesprime$.
Moreover, for each $i\in\overline{1,n-1}$ it holds that $\tau_{n-1}=\computOf{\sigma_{n-1}}$ contains a marked event $(a_i,\marking_1)$ with $\enterguess\in\marking_1$ and a marked event $(d_i,\marking_4)$ with $\leaveguess\in\marking_4$.
All other events in $\tau_{n-1}$ have empty markings.
To prove the statement for $n$, consider the possible mutual dispositions of $(i,a_n,\emptyset)$ and $(j,d_n,\emptyset)$ in $\sigma_{n-1}$.
\begin{itemize}
\item
If $(i,a_n,\emptyset)$ and $(j,d_n, \emptyset)$ are the same event, we have $\sigma_{n-1}=\sigma'\cdot(i,a_n,\emptyset)\cdot\sigma''$ and $(i,a_n,\emptyset)$ was generated by Rule~\textbf{M2}.
This transition can be replaced by~\textbf{M4} and yields $\sigma_n=\sigma'\cdot(i,a_n,\set{\enterguess,\leaveguess})\cdot \sigma''$.

\item
If $(i,a_n,\emptyset)\before{\sigma_{n-1}}(j,d_n,\emptyset)$, then $\sigma_{n-1}=\sigma'\cdot(i,a_n,\emptyset)\cdot\sigma''\cdot(j,d_n,\emptyset)\cdot\sigma'''$, where $(i,a_n,\emptyset)$ and $(j,d_n,\emptyset)$ were generated by~\textbf{M2}. These transitions can be replaced by~\textbf{M3} and~\textbf{M5} transitions, resulting in $\sigma_n=\sigma'\cdot(i,a_n,\set{\enterguess})\cdot\sigma''\cdot(j,d_n,\set{\leaveguess})\cdot
\sigma'''$.

\item
Consider $(i,a_n,\emptyset)\after{\sigma_{n-1}}(j,d_n,\emptyset)$.
With  Lemma~\ref{Lemma:GeneratesInProgramOrder} and Lemma~\ref{Lemma:GeneratesPopsImmediately}, we get $d_n\equivalenceorder a_n$.
Since barriers are not related by identity, we derive $\addrOf{d_n}\neq\ \bottom\ \neq \addrOf{a_n}$. This gives $\kindOf{a_n}=\popbkind$ and $\kindOf{d_n}=\popakind$. With Lemma~\ref{Lemma:GeneratesPopsImmediately},  $\sigma_{n-1}=\sigma'\cdot(j,d_n,\emptyset)\cdot(i,a_n,\emptyset)\cdot\sigma''$.
The events were generated by~\textbf{M2} transitions.
These transitions can be replaced by~\textbf{M6}, which yields $\sigma_n=\sigma'\cdot(j,d_n,\set{\leaveguess})\cdot(i,a_n,\set{\enterguess})\cdot\sigma''$.
\end{itemize}

Since $\sigma_n$ is obtained from $\sigma_{n-1}$ by replacing one or two marked events of rank $\rank_n$, and generation of the other events does not rely on $\markingfunc(\rank_n)$ (all other events of rank $\rank_n$ are not marked), we have $\initialwsstateprime\transitionto{\sigma_n}\stateOf{\wsautomatonprime}$ for some $\stateOf{\wsautomatonprime}\in \finalwsstatesprime$.
\end{proof}
\endlemmaatend

\begin{example}
Consider the normal-form computation $\tau\doubleprime_{\onetoone}$ (Example~\ref{Example:onetooneReshuffle}) that has the cycle~\eqref{Tag:CyclicChain} given in Figure~\ref{Figure:onetooneTrace}.
A corresponding marked computation of $\wsautomatonprime(\program, \nodecount)$ is
\begin{align*}
&(\nodeone{\writekind},\emptyset)\cdot(\nodeone{\popakind},\emptyset)\cdot(\nodetwo{\writekind},\emptyset)\cdot(\nodetwo{\popakind},\emptyset)\cdot\\
&\hspace{2cm}
(\nodeone{\barrierkind},\set{\enterguess})\cdot(\nodetwo{\barrierkind},\set{\leaveguess})\cdot(\nodeone{\loadkind},\set{\leaveguess})\cdot(\nodeone{\popbkind},\emptyset)\cdot(\nodetwo{\popbkind},\set{\enterguess}).
\end{align*}
\end{example}

Every cycle of the form \eqref{Tag:CyclicChain} has a \emph{cycle type} $\cycletype$, which is a sequence $\cycletype = \rank_{1}\ldots \rank_{k}$ of ranks from $\overline{1,\nodecount}$ with $\rank_{i}\neq \rank_{j}$ for $i\neq j$. The idea is that the events $a_i, b_i, c_i, d_i$ belong to rank $\rank_i$.
For each pair  $\rank_i, \rank_{i+1}$ in this sequence, we construct a finite automaton $\hbautomaton^{\rank_i,\rank_{i+1}}$ over the alphabet $\events\times \markings$.
It checks whether there is a conflict or identity edge from the $\leaveguess$-marked event of process $\rank_{i}$ to the $\enterguess$-marked event of process $\rank_{i+1}$.
Consider the case of conflicts. The automaton looks for a marked event $(\event_{i}, \marking_{i})$ with $\rankOf{\event_{i}}=\rank_{i}$ marked by $\leaveguess\in\marking_{i}$.
It remembers the kind and the address of this event.
Then, it seeks a marked event $(\event_{i+1},\marking_{i+1})$ with $\rankOf{\event_{i+1}}=\rank_{i+1}$ marked by $\enterguess\in\marking_{i+1}$.
If both events are found, they touch the same address, and one of them is a write, the automaton reaches the accepting state.
Since finite automata are closed under intersection, we can define the \emph{finite automaton of cycle type $\cycletype$} as
$\hbautomaton^{\cycletype}:=\hbautomaton^{\rank_1, \rank_2}\cap\ldots\cap \hbautomaton^{\rank_{k-1}, \rank_k}\cap \hbautomaton^{\rank_k, \rank_1}$.

\lemmaatend
Now we formally define the automaton $\hbautomaton^{\rank_1,\rank_2}$ that checks whether there is a conflict edge from the $\leaveguess$-marked event of process $\rank_1$ to the $\enterguess$-marked event of process $\rank_2$.
We define $\hbautomaton^{\rank_1,\rank_2}:=(\hbstates,\events\times\markings,\hbtransitions,\initialhbstate,\finalhbstates)$.
The set of states $\hbstates:=\set{\hbinit,\hbaccept}\cup(\kinddomain\times\ranks\times\addrdomain)$.
The initial state is $\initialhbstate:=\hbinit$.
The set of final states is $\finalhbstates:=\set{\hbaccept}$.
The transition relation $\hbtransitions$ is defined as follows:

\begin{description}
\item[HB1] $\hbinit\transitionto{\event,\marking}\hbinit$ with $\rankOf{\event}\neq\rank_1$ or $\enterguess\not\in{}\marking$.
\item[HB2] $\hbinit\transitionto{\event,\marking}(\kindOf{\event}, \addrOf{\event})$ for $\kindOf{\event}\neq \barrierkind$ if $\rankOf{\event}=\rank$ and $\leaveguess\in\marking$.
\item[HB3] $(\kind,\rank,\anaddr)\transitionto{\event,\marking}(\kind,\rank,\anaddr)$ for $\kind\neq \barrierkind$ if $\addrOf{\event}\neq(\rank, \anaddr)$ or $\kindOf{\event}\not\in\set{\storekind,\popbkind}$.
\item[HB4] $(\kind,\rank,\anaddr)\transitionto{\event,\marking}(\hbaccept)$ for $\kind\neq \barrierkind$ if $\addrOf{\event}=(\rank,\anaddr)$, $\rankOf{\event}=\rank_2$, $\enterguess\in\marking$, and $\set{\kind,\kindOf{\event}}\cap\set{\storekind,\popbkind}\neq\emptyset$.
\item[HB5] $\hbaccept\transitionto{\event,\marking}\hbaccept$ for all $(\event,\marking)\in\events\times\markings$.
\item[HB6] $\hbinit\transitionto{\event,\marking}\hbbarrier$ if $\kindOf{\event}=\barrierkind$, ($\rankOf{\event}=\rank_1$ and $\leaveguess\in\marking$) or ($\rankOf{\event}=\rank_2$ and $\enterguess\in\marking$).
\item[HB7] $\hbbarrier\transitionto{\event,\marking}\hbbarrier$ if $\kindOf{\event}=\barrierkind$, $\rankOf{\event}\not\in\set{\rank_1,\rank_2}$.
\item[HB8] $\hbbarrier\transitionto{\event,\marking}\hbaccept$ if $\kindOf{\event}=\barrierkind$, ($\rankOf{\event}=\rank_1$ and $\leaveguess\in\marking$) or ($\rankOf{\event}=\rank_2$ and $\enterguess\in\marking$).
\end{description}

\begin{lemma}\label{Lemma:HBAutomatonWorks}
Consider $\rank_1,\rank_2\in\ranks$ and $\tau\in\langOf{\wsautomatonprime(\program, \nodecount)}$ that has a single marked event $(\event_i,\marking_i)$ with $\leaveguess\in\marking_i$ and $\rankOf{\event_i}=\rank_1$ and a single $(\event_j,\marking_j)$ with $\enterguess\in\marking_j$ and $\rankOf{\event_j}=\rank_2$.
Then $\tau\in\langOf{\hbautomaton^{\rank_1,\rank_2}}$ iff $\event_i\cycorder\event_j$.
\end{lemma}
\begin{proof}
We give the proof for memory accesses, the argumentation in the case of barriers is similar.
We start with the implication from left to right.
In order to reach the accepting state $\hbaccept$ the first time, the automaton must have reached a state $(\kind,\rank,\anaddr)$ and performed a transition defined by~\textbf{HB4}.
This transition had to consume the symbol $(\event_j,\marking_j)$ which is, according to the statement of the lemma, the only marked event in $\tau$ with $\rankOf{\event_j}=\rank_2$ and $\enterguess\in\marking_j$.
The state $(\kind,\rank,\anaddr)$ was reached the first time via a transition defined by~\textbf{HB2}.
This transition had to consume the symbol $(\event_i,\marking_i)$ which is, according to the statement of the lemma, the only marked event in $\tau$ with $\rankOf{\event_i}=\rank_1$ and $\leaveguess\in\marking_i$.
According to \textbf{HB2}, $\kind=\kindOf{\event}$ and $(\rank,\anaddr)=\addrOf{\event}$.
Therefore, \textbf{HB4} requires that $\event_i$ and $\event_j$ access the same address and at least one of them is a write.
Moreover, according to Rule~\textbf{HB3}, the automaton could not consume a marked event which is a write to $(\rank,\anaddr)$ after reading $(\event_i,\marking_i)$ and before reading $(\event_j,\marking_j)$.
Altogether, by definition of the conflict relation, $\event_i\conflictorder\event_j$.

For the proof from right to left,
let $\tau=\tau_1\cdot(\event_i,\marking_i)\cdot\tau_2\cdot(\event_j,\marking_j)\cdot\tau_3$.
The first part, $\tau_1$, is read by the transitions defined by~\textbf{HB1}.
Indeed, $(\event_i,\marking_i)$ is the only marked event in $\tau$ that does not satisfy the requirements of this rule.
Then the automaton performs the transition defined by~\textbf{HB2}, reads $(\event_i,\marking_i)$, and reaches the state $(\kind, \rank, \anaddr)$ with $\kind=\kindOf{\event}$ and $(\rank,\anaddr)=\addrOf{\event}$.
Since $\event_i\conflictorder\event_j$, part $\tau_2$ does not contain writes to $\addrOf{\event_i}$. It is consumed by the transitions defined by~\textbf{HB3}.
Finally, the automaton performs the transition defined by~\textbf{HB4} and reaches the accepting state.
There it loops on the symbols from $\tau_3$.
\end{proof}

\begin{lemma}\label{Lemma:ComputationFromIntersectionIsViolating}
Consider a cycle type $\cycletype$ and let $\tau\in\projectionOf{\big(\langOf{\wsautomatonprime(\program, \nodecount)}\cap
\langOf{\hbautomaton^\cycletype}\big)}{\events}$.
Then $\tau$ is a computation of $(\program, \nodecount)$ and has a cyle~\eqref{Tag:CyclicChain} of type $\cycletype$.
\end{lemma}
\begin{proof}
By Lemma~\ref{Lemma:WSAutomatonGeneratesValidComputations} and Lemma~\ref{Lemma:MarkedWSAutomatonWorksAsBefore}, $\tau$ is a computation of program $(\program, \nodecount)$.
By Lemma~\ref{Lemma:MarkedWSAutomatonGuesses} and Lemma~\ref{Lemma:HBAutomatonWorks}, $\tau$ has a dependence chain~\eqref{Tag:CyclicChain} of type $\cycletype$.
\end{proof}

\begin{lemma}\label{Lemma:ViolatingWellShapedComputationIsInIntersection}
Consider a cycle type $\cycletype$ and let $\tau$ be a normal-form computation of $(\program, \nodecount)$ that has a cycle~\eqref{Tag:CyclicChain} of this type.
Then $\tau\in\projectionOf{\big(\langOf{\wsautomatonprime(\program, \nodecount)}\cap\langOf{\hbautomaton^{\cycletype}}\big)}{\events}$.
\end{lemma}
\begin{proof}
By Lemma~\ref{Lemma:MarkedWSAutomatonGuessesEverything}, $\wsautomatonprime(\program, \nodecount)$ can generate $\tau'$ with $\projectionOf{\tau'}{\events}=\tau$, the events $a_i,d_i$ from~\eqref{Tag:CyclicChain} marked by $\enterguess$ and $\leaveguess$ respectively, and the other events marked by $\emptyset$.
By Lemma~\ref{Lemma:HBAutomatonWorks}, the automata $\hbautomaton^{\rank_i,\rank_{i+1}}$ will accept $\tau'$, due to $d_i\cycorder{}a_{i+1}$.
\end{proof}
\endlemmaatend

\begin{theorem}\label{Lemma:CheckingIntersectionEmptinessIsEnough}
$\program$ is robust iff $\langOf{\wsautomatonprime(\program, \nodecount)}\cap\langOf{\hbautomaton^{\cycletype}}=\emptyset$ for all cycle types $\cycletype$.
\end{theorem}
\proofatend
The statement follows from Theorem~\ref{Theorem:OnlyNeedWellShapedComputations}, Lemma~\ref{Lemma:ComputationFromIntersectionIsViolating}, Lemma~\ref{Lemma:EachRankContributesOnce}, and Lemma~\ref{Lemma:ViolatingWellShapedComputationIsInIntersection}.
\endproofatend

We can now prove Theorem~\ref{th:main}.
To check whether $(\program, \nodecount)$ is robust, we go over all cycle types $\cycletype=\rank_1\ldots \rank_k$.
This enumeration of cycle types can be done in space that is polynomial in $\nodecount$.
For each such sequence, we check if $\langOf{\wsautomatonprime(\program,\nodecount)}\cap\langOf{\hbautomaton^\cycletype}=\emptyset$.
By Theorem~\ref{Lemma:CheckingIntersectionEmptinessIsEnough}, the program is robust iff all intersections are empty.
By Lemma~\ref{Lemma:MultiheadedIntersection}, there is a $4$-headed finite state automaton $W$ with $\langOf{W}=\langOf{\wsautomatonprime(\program, \nodecount)}\cap\langOf{\hbautomaton^\cycletype}$.
Since the size of $W$ is exponential in the size of $(\program, \nodecount)$ and emptiness is in \nl\ by Lemma~\ref{Lemma:EmptinessComplexity}, deciding $\langOf{W}=\emptyset$ can be done in space that is polynomial in $(\program, \nodecount)$.
This shows robustness is in $\pspace$.

\begin{spacing}{0.9}
\bibliographystyle{plain}
\bibliography{cited}
\end{spacing}

\clearpage
\appendix\appendixtrue
\section{Missing Proofs}

For some of the following proofs, we assume that Table~\ref{Table:EventAutomatonRulesFull} and Table~\ref{Table:WellShapedAutomatonRulesFull} associate with each event $e$ the transition in the program that produced this event: $\instructionOf{e}$.
Also, for a $\readkind$, $\writekind$, $\popakind$, or $\popbkind$ event we write $\queueOf{e}$ to denote the id of the queue being modified by this event $\event$.

\printproofs

\end{document}